\documentclass[a4paper,twocolumn,11pt]{quantumarticle}
\pdfoutput=1
\usepackage[utf8]{inputenc}
\usepackage[english]{babel}
\usepackage[T1]{fontenc}
\usepackage{hyperref}

\usepackage{tikz}

\usepackage{graphicx} 
\usepackage{amsfonts,amsmath,amssymb,amsthm}
\newtheorem{theorem}{Theorem}

\newtheorem{definition}{Definition}
\newtheorem{corollary}{Corollary}
\newtheorem{remark}{Remark}
\usepackage{xcolor}
\usepackage{hyperref}
\usepackage{cite}
\usepackage{enumitem}
\usepackage{comment}
\usepackage{mathtools}
\usepackage{subcaption}
\newtheorem{lemma}[theorem]{Lemma}

\DeclarePairedDelimiter\floor{\lfloor}{\rfloor}

\begin{document}

\title{Growing Sparse Quantum Codes from a Seed}

\author{ChunJun Cao}
\email{cjcao@vt.edu}
\affiliation{Department of Physics, Virginia Tech, Blacksburg, VA, USA 24061}
\affiliation{Virginia Tech Center for Quantum Information Science and Engineering, Blacksburg, VA 24061, USA}

\author{Brad Lackey}
\email{Brad.Lackey@microsoft.com}

\affiliation{Microsoft Quantum, Redmond, WA, USA}
\maketitle

\begin{abstract}

It is generally unclear whether smaller codes can be ``concatenated'' to systematically create quantum LDPC codes or their sparse subsystem code cousins where the degree of the Tanner graph remains bounded while increasing the code distance. In this work, we use a slight generalization of concatenation called conjoining introduced by the quantum lego formalism. We show that by conjoining only quantum repetition codes, one can construct quantum LDPC codes. More generally, we provide an efficient iterative algorithm for constructing sparse subsystem codes with a distance guarantee that asymptotically saturates $kd^2=O(n)$ in the worst case. Furthermore, we show that the conjoining of even just two-qubit quantum bit-flip and phase-flip repetition codes is quite powerful as they can create any CSS code. Therefore, more creative combinations of these basic code blocks will be sufficient for generating good quantum codes, including good quantum LDPC codes.
\end{abstract}

\section{Introduction}

Concatenation is a  simple yet powerful tool for coding theory. Because these codes are grown from smaller seed codes in a straightforward manner, many properties such as distance, rate, enumerators, and fault-tolerant (FT) threshold can be rigorously controlled and computed \cite{Knill:1996ex}. A common textbook example of concatenated code refers to tree-like concatenation, where the resulting code after $L$ layers of concatenation is $[[n^L,1,d^L]]$. However, more general concatenation involves using codes at various rates at different layers, such as \cite{Yamasaki_2024,yoshida2024concatenate,gidneyconcat} and also geometric connectivity, such as the holographic codes\cite{Pastawski_2015,Harris_2018,ABSC,evenblycode,biasedholo}. These codes carry significant practical interest as they are conceptually simple yet contain a wide variety of architectures. They are also efficiently decodable and can produce competitive spacetime overheads for fault tolerance \cite{Yamasaki_2024,yoshida2024concatenate,gidneyconcat}. In particular, it is much easier to control or design the code properties such as code distance and (addressable) fault-tolerant gates\cite{yoder_piece,hetero_tree1,cao_lackey_targeted,Steinberg:2025kzb,Ferris_2014}.  It is also much easier to simulate or analyze the performance of such codes thanks to their modularity\cite{Litinski:2025irj}.

However, a common drawback for concatenated codes is their syndrome check weight, which can grow exponentially with the number of layers $L$. This is unavoidable in standard concatenation as it involves contracting a sequence of isometric tensors in the tensor network representation of quantum codes. As the outer code increases distance, it also expands the checks of inner code stabilizers. In contrast, sparse quantum codes, i.e., codes where check weights and the number of checks each qubit talks to is constant have remained the leading contenders to fault-tolerant quantum computation thanks to their lower overhead scaling and sparse spatial connectivity\cite{pattison_qldpc,hayata_qldpc,gottesmanldpc_overhead,Bravyi_2024,panteleev2022asymptoticallygoodquantumlocally,Tillich_2014,qldpc_rev}. This includes topological codes, more general quantum low-density parity-check (LDPC) codes, and their sparse subsystem code cousins. In addition to these direct constructions, methods such as \cite{hastings2016,hastings2023,Vasmer2024,sparsecodecircuit,baspin2024wirecodes} can be applied to sparsify codes that may otherwise carry large weight checks. However, while many above constructions work well in asymptopia by examining the scaling of code parameters, the constant scaling can still be wildly infeasible in practical quantum computing. Furthermore, efficient decoding, simulation, and the addressability of (non-Clifford) multi-qubit FT gates in these codes can also present separate challenges.

It is therefore natural to ask whether one can combine the advantages from both ideas above to grow finite size quantum LDPC or sparse subsystem codes with reasonably good parameters directly from smaller code modules in a way that can preserve desirable FT properties. This is interesting both coding-theoretically and practically. There are many method of code construction using smaller modules that do not directly correspond to concatenation, e.g., quantum Tanner codes\cite{leverrier2022quantumtannercodes}, ZX calculus \cite{ZXcalc,kissinger2022phasefreezxdiagramscss}, holographic codes, and quantum lego codes\cite{QL1,Farrelly_2021,Farrelly_2022,QL2}. The last two are particularly intriguing because unlike the other formalisms, suitable choice of the lego blocks also provides a straightforward way to produce codes that support FT non-Clifford gates and even targeted gates\cite{cao_lackey_targeted,QL1} through operator pushing. The quantum lego formalism also demonstrated that any quantum code, including sparse codes, can be written as the combination of a few types of atomic lego blocks\cite{QL1,QL2,xpql}. 

However, in both concatenation and the quantum lego generalization, we still lack an explicit pathway to synthesize sparse codes with target code parameters from smaller atoms. For example, the proof of expressivity in \cite{QL1,xpql} are existence proofs. For stabilizer codes \cite{kissinger2022phasefreezxdiagramscss,QL2}, the proof is slightly more constructive where it provides an explicit construction as atomic codes once the stabilizers are known. Nevertheless it does not provide a way to produce those stabilizers in the first place if we are only given simpler code parameters like rate or distance. The tensor network produced is also often far from efficient. Thus far, only isolated examples have been produced as a proof of principle \cite{QL1,Farrelly_2022,QL2} based on intuition, but no generalizing principle has been distilled. The same is true for ZX-based approaches where it is often easy to reverse engineer a code diagrammatically, but hard to build one from scratch without prior knowledge.

Here we take a first step filling in these gaps of knowledge above where we provide a generalizable principle for growing sparse codes in a fashion similar to concatenation, such that one can add to the code module by module to increase distance, but doing so without ruining sparsity. Although we do not discuss FT gates in the current work, it is a necessary step for the eventual goal for growing qLDPC codes with useful addressable FT gates. Explicitly, we use the conjoining operations introduced by quantum lego\cite{QL1,QL2} to glue together smaller code fragments. Diagrammatically, these are simple and natural extensions of code concatenation and their operation over check matrices have been defined in \cite{QL1}. One can also think of conjoining simply as concatenation but the ``code'' we concatenate with can be pathological, i.e., its encoding map can be non-isometric. 
We first motivate the general principle behind conjoining quantum lego blocks that can preserve sparsity. Then we specialize to conjoining smaller lego blocks for the sake of simplicity, in particular, 2-qubit bit-flip and phase-flip repetition codes. We show that these elements are already sufficiently universal such that any CSS code can be produced from an abundant supply of them. In a more constructive manner, we first shows how LDPC code like the surface code or compass code can be built by alternating regular concatenation and conjoining of quantum repetition codes. Then we introduce a simple algorithm that constructs sparse subsystem codes iteratively where the resulting code parameter achieves $kd^2=O(n)$ asymptotically in the worst case. In constructing the explicit examples, we also introduce a new tensor network for the 2d surface code and 2d compass code which may be of independent interest. We also show familiar examples like the Bacon-Shor code can emerge from this algorithm while also introducing new code constructions using this method as a proof of principle. 

We briefly review the basics of quantum lego in Sec.~\ref{sec:2} then discuss the conjoining of LDPC and sparse codes in Sec.~\ref{sec:qldpc} and \ref{sec:4}. In Sec~\ref{sec:5}, we explain why the use of repetition codes as atomic legos is sufficient to cover CSS codes in generality, and finally conclude with remarks on connections with existing and future work in Sec~\ref{sec:discussion}.

\section{Growing codes graphically with quantum lego}\label{sec:2}
Traditionally, the default method for growing a large code with good property from smaller code blocks is code concatenation, where each physical qubit of an inner code is encoded as the logical qubit of outer codes, see for example Figure~\ref{fig:genshor} below. There each branch of the tree represents a separate encoding map of the inner or outer code. The red and green node correspond to phase-flip and bit-flip repetition codes (or $X$- and $Z$-spiders) respectively. More generally, the choices of inner and outer codes can vary and the encoding can become complicated to track. 

Therefore, it is often far more effective to use a graphical representation called tensor network (or ZX diagram\footnote{Recently, it is becoming more common to use ZX diagrams to describe QECCs. In the context of this work, we can simply treat them as tensor networks with specialized tensors that obey particular transformation rules. }) to describe the construction of quantum codes\cite{Ferris_2014,Farrelly_2021,QL1,Farrelly_2022,QL2,xpql,de_Beaudrap_2020,Chancellor_2023,kissinger2022phasefreezxdiagramscss,Kissinger_2024,Huang_2023}. The specific approach we focus on here is called Quantum Lego (QL)\cite{QL1,QL2,xpql,tensor_enum,QLRL}, which is a more general framework that designs codes by combining smaller codes to form big codes. In particular, it was shown that QL can construct any code by joining certain sets of smaller quantum codes \cite{QL1}. Additionally, QL applies to stabilizer and non-stabilizer codes by keeping track of their symmetries that fix the fault-tolerant gates supported by the code\cite{xpql}. More detailed information of the framework can be found in\cite{QL1,xpql}. 

A ``lego block'' in QL is a tensor, which can be represented graphically as a node with dangling edges, where each edge or leg maps to a tensor index (Figure~\ref{fig:qlreview}a). Physically, the tensor $V_{i_1,i_2,\dots}$ can attain different physical meanings when we contract different bases. The interconversion between these different objects represented by the same tensor captures the Choi-Jamiolkowski isomorphism and code shortening\cite{gottesman1997stabilizer,raissi_modifying,QL1} (Figure~\ref{fig:qlreview}a) which identify the dualities between maps $V$ and states $|V\rangle$. 

For example, consider a 4-index tensor 
$$V_{ijk\ell}=\begin{cases}
    1\quad i=j=k=\ell\\
    0\quad \mathrm{otherwise}
\end{cases}$$
where $i,j,k,\ell=0$ or $1$. By designating certain legs/indices of the same tensor as input or outputs of a map, it can represent a 4-qubit GHZ state $|0000\rangle+|1111\rangle = V_{ijk\ell}|i,j,k,\ell\rangle$, an encoding isometry of a 3-qubit repetition code $|000\rangle\langle 0|+|111\rangle\langle 1|=V_{ijk\ell}|ijk\rangle\langle \ell|$, a projection $|00\rangle\langle 00|+|11\rangle\langle 11|=V_{ijk\ell}|ij\rangle\langle k\ell|$ onto the space spanned by $\{|00\rangle,|11\rangle\}$, or a linear functional $|0\rangle\langle 000|+|1\rangle\langle 111|=V_{ijk\ell}|i\rangle\langle jk\ell|$, where in each of these repeated indices are summed over. This means that the same tensor can be used for different purposes. For example, a tree-like contraction of such tensors in Figure~\ref{fig:genshor} is a composition of encoding maps of repetition codes and therefore corresponds to code concatenation. However, the same tensor can also be used such that more than one leg is designated as inputs. As tensors can be joined together without having to worry about the directionality of the map, it is far more convenient to describe the construction of code with (encoding) tensors $V_{i_1i_2,\dots}$ first, then only assign directionality to the maps that tensors can describe when needed. Often, it is helpful then to work directly with the states $|V\rangle$ they describe, then convert it into an encoding map after all gluing have been completed.

In practice, one can obtain these tensors from the encoding map of a small quantum codes or particular quantum states with symmetries. One can also custom design suitable tensors through a structured search, e.g. using variational algorithms, as the total system size is small\cite{VQAQEC}. Tensors constructed this way are desirable as they inherit a large number of symmetries from the fault-tolerant logical gates of a code or from the stabilizer group of a state. On the one hand, we can take advantage of these symmetries to simplify code design --- instead of working with exponentially many tensor components, one can track these symmetries more efficiently to characterize underlying ``quantum lego block''. For example, a stabilizer state over $n$ qubits is equipped with a stabilizer group with $n$ generators. As such, one only needs $O(n^2)$ instead of $O(2^n)$ parameters to uniquely fix the state or its tensor coefficient. 
In this case, each tensor, or rather, its dual state $|V\rangle$, can be fully characterized by its Pauli stabilizer group or equivalently its check matrix in the symplectic representation.

On the other hand, these symmetries can be converted back into fault-tolerant gates of a code as needed. For example, consider a symmetry of a tensor of the form $\mathcal{O}_1\otimes \mathcal O_2\otimes \mathcal O_3\otimes \mathcal O_4\otimes \mathcal O_5$ as in Figure~\ref{fig:qlreview}b. Note that this operator need not be Pauli. One can then assign different physical meanings to this symmetry in different aspects of the duality. As a state, it means that $\mathcal O_1\otimes \mathcal O_2\otimes \mathcal O_3\otimes \mathcal O_4\otimes \mathcal O_5|V\rangle=|V\rangle$, i.e. it stabilizes the state. But it also can be interpreted as a logical operator $\bar{\mathcal O_1^t}=\mathcal O_2\otimes \mathcal O_3\otimes \mathcal O_4\otimes \mathcal O_5$ of a code with encoding map $V$ if we have chosen leg $i_1$ to be the input logical leg. Generally, for tensors obtained from Pauli stabilizer codes, the associated stabilizer and normalizer fully and efficiently fixes the underlying tensor without ever having to deal with its components. A similar but more nuanced generalization applies to XP stabilizer codes also.

\begin{figure*}
    \centering
    \includegraphics[width=0.81\linewidth]{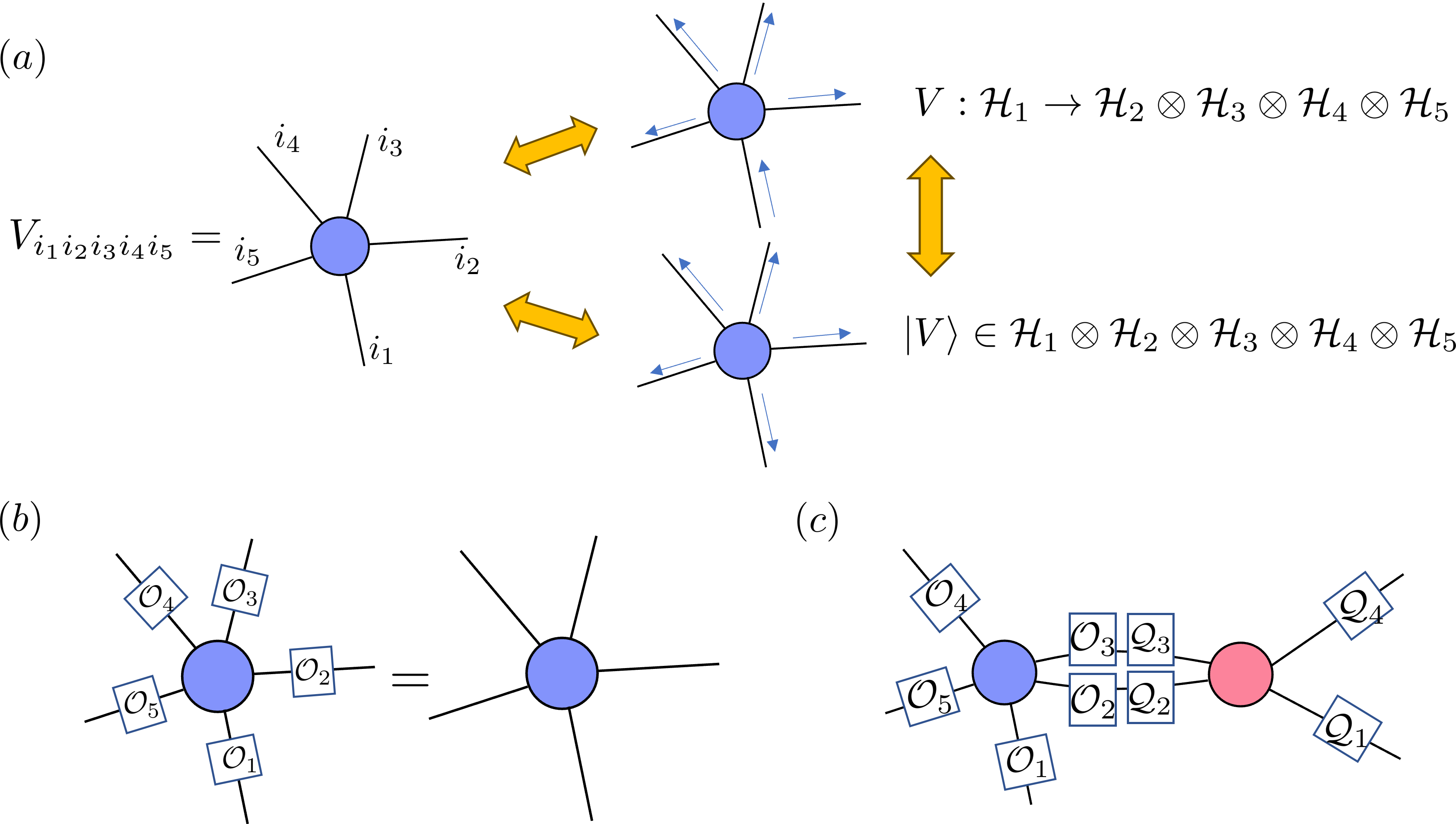}
    \caption{(a) A tensor can be represented as a state or a map, depending on how one assigns the meaning of the indices or legs. This can be understood as a graphical manifestation of the Choi-Jamiolkowski isomorphism. (b) A tensor has a symmetry if its contraction with other tensors leave the original tensor invariant. One can think of a symmetry as a stabilizer of the state that the tensor represents. (c) When two tensors are glued together, the symmetries of the larger tensor network can be obtained via operator matching. Here the operators match if $\mathcal O_3=\mathcal Q_3^*$ and $\mathcal O_2=\mathcal Q_2^*$. If the matching condition is satisfied, then the operators acting on the dangling legs remain a symmetry of the tensor network.}
    \label{fig:qlreview}
\end{figure*}

To build bigger codes, these ``quantum lego blocks'' can then be glued together. Such gluing operations are natural generalization of code concatenation from a graphical perspective. While concatenation is gluing with certain restrictions on the orientation of the tensors,  those restrictions are completely revoked in general tensor contractions. To characterize the bigger codes, we generate their symmetries from their constituents through ``operator matching.'' For example as in Figure~\ref{fig:qlreview}c, if $\mathcal O_1\otimes \mathcal O_2\otimes\dots$ and $\mathcal Q_1\otimes \mathcal Q_2\otimes \dots$ are symmetries of two legos, and we contract along legs $2$ and $3$, then these must match on the connecting edges in order for the operator $\mathcal O_1\otimes \mathcal O_4\otimes \mathcal O_5\otimes \mathcal Q_1\otimes \mathcal Q_4$ to remain a symmetry of the bigger tensor network. In this context we say two operators $\mathcal O,\mathcal Q$ match if $\mathcal O=\mathcal Q^*$. Repeating this procedure for all symmetries on each lego block, we generate the symmetries of the bigger codes. For stabilizer codes, they are also sufficient to uniquely determine the resulting codes. An equivalent, but convenient, description of operator matching is operator pushing, where one cleans any operator acting on connected edges with local symmetries of the tensor so that the non-trivial operators only act on the dangling edges. For example, taking the top leg of Figure~\ref{fig:genshor} as the logical qubit, and the remaining leg as physical qubits, logical and stabilizer operators can be visualized as operators flowing through the network\footnote{The operator needs to change into its complex conjugate when it flows through an edge. For Paulis, the complex conjugate is always equal to the operator itself up to a global phase, hence we only need to track the Pauli type.}. 
 \begin{figure}
     \centering
     \includegraphics[width=1\linewidth]{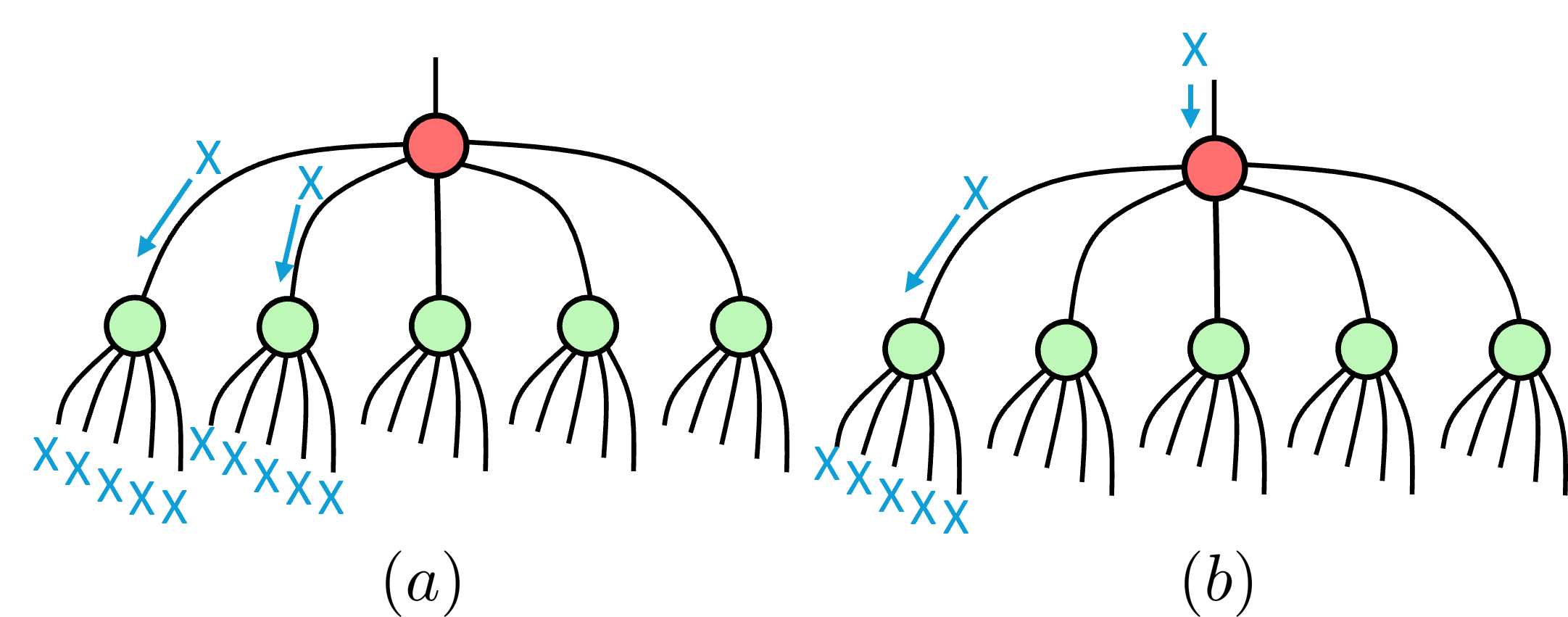}
     \caption{A $[[25,1,5]]$ generalized Shor code by gluing together $X$- and $Z$-spiders. The top leg represents the logical input while the remaining dangling legs represent the physical qubits. (a) A stabilizer of the code can be obtained from pushing operators that act as the identity on the logical leg. (b) A representation of the logical X operator can be obtained by pushing operator X through the logical leg to the physical legs, generating an operator flow. (c.f. Pauli flow\cite{ZXcalc}.)}
     \label{fig:genshor}
 \end{figure}
 
To generate a stabilizer, we first write down a stabilizer of a seed code and then push it to the dangling legs. For example, first put down the inner code stabilizer $XX$ acting on its tensor (Figure~\ref{fig:genshor} left). Cleaning, or pushing the $X$ off of the connecting edges using symmetries of the green tensors, we add $X^{\otimes 10}$ on ten of its dangling physical legs. A similar procedure can be repeated to produce a logical operator (Figure~\ref{fig:genshor} right). It is clear that in this graphical language, we can simply represent the concatenation of single qubit ($k=1$) codes as trees where each tensor is obtained from the encoding isometry of the inner/outer codes. The tensors and their connectivity then fully specify the details of the concatenation. For block codes ($k>1$), a richer network connectivity can follow, allowing for other geometries to emerge, e.g. holographic codes.

While graphical rules are intuitive and useful for visualization, for practically efficient implementation on a computer, it is often convenient to describe tensor gluing of stabilizer codes using an equivalent check matrix operation called \emph{conjoining}. Since tensors obtained from (non-Abelian) stabilizer codes can be specified by the corresponding check matrices of their dual states $\{|V_k\rangle\}$, each diagrammatic gluing operation can be equivalently described by a conjoining operation $\wedge$ over the check matrices of $\{|V_k\rangle\}$. This is defined for the check matrices of Pauli stabilizer codes/states over qudits of prime dimension in the Appendix of \cite{QL1}. Interestingly, Pauli stabilizer codes are closed under conjoining, i.e., gluing together Pauli stabilizer codes/states always produce another Pauli stabilizer code/state. Similar conjoining operations can also be generalized to certain non-Abelian stabilizer codes\cite{xpql} where the check matrices are no longer over finite fields, but the closure property is no longer satisfied.

\section{qLDPC from Conjoining}\label{sec:qldpc}
In the following, we will focus on the gluing/conjoining of Pauli stabilizer codes over qubits\footnote{Since we only consider stabilizer codes, we will use gluing and conjoining interchangeably.}. In fact, for simplicity we will only be conjoining two types of codes: the bit-flip and phase-flip repetition codes. Despite this simplicity, such concatenated codes have been able to produce impressive results\cite{Yamasaki_2024, yoshida2024concatenate, gottesman2009introductionquantumerrorcorrection}, but there is much room for improvement. In particular, for the purposes of the paper, the support of the stabilizer checks under concatenation grow with the number of layers; such large weight operators are not only experimentally costly to implement for error detection, but also potential sources for spreading errors, unless FT schemes like Shor/Steane/Knill error correction are used that in turn incur additional overhead associated with ancilla state preparation.

An obvious solution to this problem is to consider qLDPC such that the Tanner graph of the code has bounded degree so that each check only acts on a constant number of qubits, and each qubit is only checked by a constant number of checks. It was shown that LDPC codes permit efficient overhead scaling\cite{gottesmanldpc_overhead} if the code has linear distance and rate\cite{leverrier2022quantumtannercodes, goodldpc2, Panteleev_2022, pattison_qldpc, hayata_qldpc}. However, unlike concatenated codes, it is unclear how qLDPC codes or more general sparse codes can be grown using smaller codes through a concatenation process. 

From the point of view of operator flow, it is clear that check weights do not typically decrease under concatenation if we also want the code distance to increase. Intuitively to get sparse codes, we aim to create operator flows wherein the support of weights are reduced generally, but without damaging the support of the minimal weight logical operators. Informally, a concatenation approach that reduces check weight without reducing distance is to choose ``bad'' block codes where high weight logical operators have low distance whereas the low weight logical operators have distance at least their logical weight, then as long as logical operators and checks are oriented in such a way that the qubits on which it acts has a larger overlap with checks than with logical we would be done.

Of course this is impossible, at least for stabilizer codes. As each logical $\bar{X}_j$ has logical weight $1$, we would require it to have physical weight $1$, which leads to the trivial code. Putting weight $1$ aside, there are many distance $d=2$ codes that illustrate this. For example, consider the $[[6,4,2]]$ iceberg code\cite{self2024protecting}. We have logical operators $\bar{X}_j = X_0 X_j$ and $\bar{Z}_j = Z_jZ_5$, for $j=1,2,3,4$. Weight two logical-$X$ operators can be represented as weight two physical operators $\bar{X}_j\bar{X}_k = X_j X_k$ and similarly for logical-$Z$ operators. The same holds at weight three, e.g. $\bar{X}_1\bar{X}_2\bar{X}_3 = X_0 X_1 X_2 X_3 = X_4 X_6$. And similarly at weight four: $\bar{X}_1\bar{X}_2\bar{X}_3\bar{X}_4 = X_0 X_6$ and $\bar{Z}_1\bar{Z}_2\bar{Z}_3\bar{Z}_4 = Z_0 Z_6$.

Now suppose we have a CSS code with some large weight $X$- or $Z$-checks. Continuing our example, we aim to find a subset of physical qubits $A$, with $|A|=4$, that has significant overlap with the support of some of these large weight checks. Then we ``partially'' concatenate, in that the physical qubits in $A$ are now identified as the logical qubits of this iceberg code. The upshot is that (i) any operator that overlaps $A$ in at least $3$ places will have its weight reduced, (ii) operators that intersect $A$ in one place will have its weight increased by one, and (iii) a weight $6$ $X$-check and $Z$-check will be added to the code. Writing $w$ as our desired check weight bounds, this operation is productive if the following conditions are satisfied:
\begin{enumerate}
    \item (as after conjoining must have distance at least $d$) whenever $\bar{L}$ is a nontrivial logical operator with $|\mathrm{supp}(\bar{L})\cap A| = 3+a$ then $\mathrm{wt}(L) > d + a$ pre-trace;

    \item (as after conjoining we must satisfy the check bound $w$) whenever $S$ is a check operator with $|\mathrm{supp}(S)\cap A| = 1$ then $\mathrm{wt}(S) < w$ pre-trace; and,

    \item $w \geq 6$.
\end{enumerate}
Upon conjoining, any check operator $S$ with $|\mathrm{supp}(S) \cap A| \geq 3$ will have its weight reduced in the new code. Additionally,
\begin{itemize}
    \item by (1) above, the weight of the conjoined code remains at least $d$;
    \item by (2) above, any check operators with $\mathrm{wt}(S) \leq w$ will remain so; and,
    \item by (3) above, all new check operator have weight $\leq w$.
\end{itemize}

The key property of the iceberg code is that after some threshold $t$, any logical operator $\overline{L}$ with logical Hamming weight $\overline{\mathrm{wt}}(\overline{L}) \geq t$ can be represented by a physical operator of Hamming weight $\mathrm{wt}(\overline{L}) \leq t$. For example, the max rate holographic pentagon code also satisfies this for some choice of the bulk and boundary qubits, but so do many codes where $k>d$.

After that, it is just a matter of aligning the set of legs $A$ so that logical operators of the original code of weight $\geq t$ do not have their weight reduced below $d$, and stabilizer of weight $\leq t$ do not have their weight increased above $w$. Naturally, new checks introduced by the conjoined code must have weight $\leq w$.

Unfortunately, large codes like this are difficult to work with and align with appropriate subsets of legs, which is important for our strategy. Instead we start with the simplest codes given by the fusion of two $2$-qubit repetition codes (equivalently $X$- and $Z$-spiders)  as in Figure~\ref{fig:zxspider}. These could also be obtained from the encoding tensors of $3$-qubit repetition codes. We will define the $4$-legged green (respectively red) tensor as a non-isometric code of $Z$-type (respectively $X$-type) when two of its legs are used as input/output. Strictly speaking, they are not codes as their encoding map contains a kernel, and hence operator pushing is therefore non-isometric\footnote{Nonetheless if one where to mod out the kernel one would recover a usual repetition code.}. 

These non-isometric codes satisfy our previous criteria. The $Z$-type code (i) maps high logical weight (namely weight two) $Z$-operators to the identity, (ii) preserves low logical weight as weight one $Z$-operators push to weight one operators, and (iii) adds one weight-4 $X$ symmetry. It can be shown that this symmetry will preserve the weight of any $X$ checks $S$ such that $|\mathrm{supp}(S)\cap A|=2$ whereas it will merge any two $X$ checks $S$ and $S'$ where $|\mathrm{supp}(S)\cap A|=|\mathrm{supp}(S')\cap A|=1$ into one bigger check $S''$ with weight $\mathrm{wt}(S'')=\mathrm{wt}(S)+\mathrm{wt}(S')$. Similarly, the $X$-type code has the same properties with the role of $X$ and $Z$ reversed.

Interestingly, objects of such kind do occur in the construction of topological codes using quantum lego \cite{QL1,QL2}. Therefore, we focus on these simple components and see where it takes us\footnote{In a crude sense, they may be thought of as a $d=0$ code that encodes 2 qubits as the logical $\bar{Z}\bar{Z}$ operator is formally mapped to the identity operator, which has weight $0$.}.
Although we use non-isometric codes here, we do not claim that they are strictly necessary. Indeed, if one considers an encoding map of a CNOT, it also satisfies the property we ask for, as do other codes where $d<k$. Yet, we do not want to expand any checks that only overlaps with one qubit. As such, the simplest component we can ask for is a code that does not expand any flow with $1$ input, but reduces a flow that has more than $2$ inputs. The smallest such example is a rank $4$ tensor, and by our restriction, precisely the $4$-qubit GHZ state we are using in Figure~\ref{fig:zxspider}b. 

\begin{figure}
    \centering
    \includegraphics[width=0.8\linewidth]{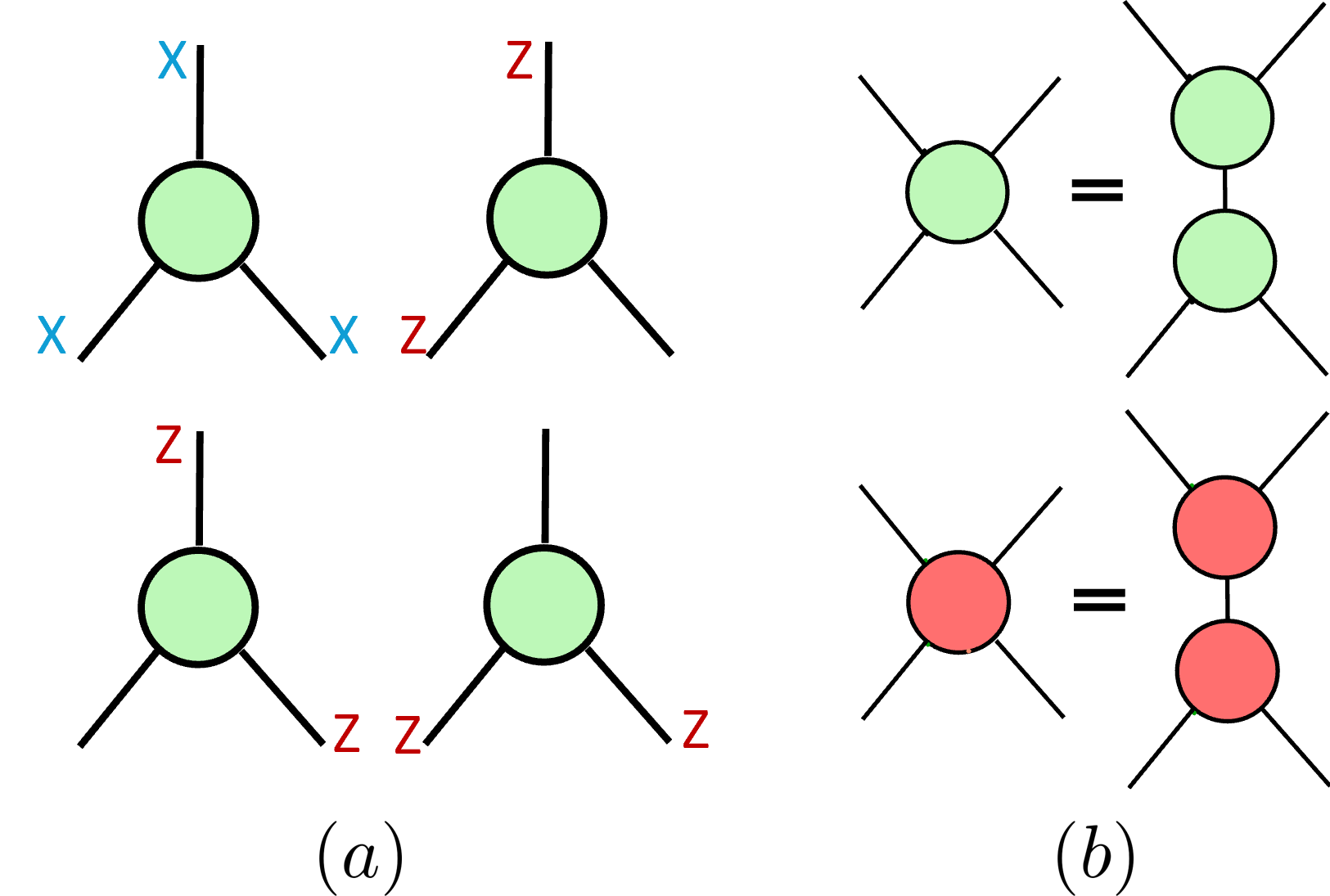}
    \caption{(a) Z spiders (green) of valence 3 are represented as rank 3 tensors which has a number of symmetries generated by the Pauli operators. For X spiders (red) of the same type, one swap the roles of X and Z and the same symmetries apply. (b) These spiders can be merged to form larger spiders, which are precisely our X and Z type non-isometric tensors. }
    \label{fig:zxspider}
\end{figure}

\subsection{Example: Surface code}
A simple non-trivial example of a qLDPC code that this protocol produces is the rotated surface code.
\begin{figure}
    \centering
    \includegraphics[width=0.95\linewidth]{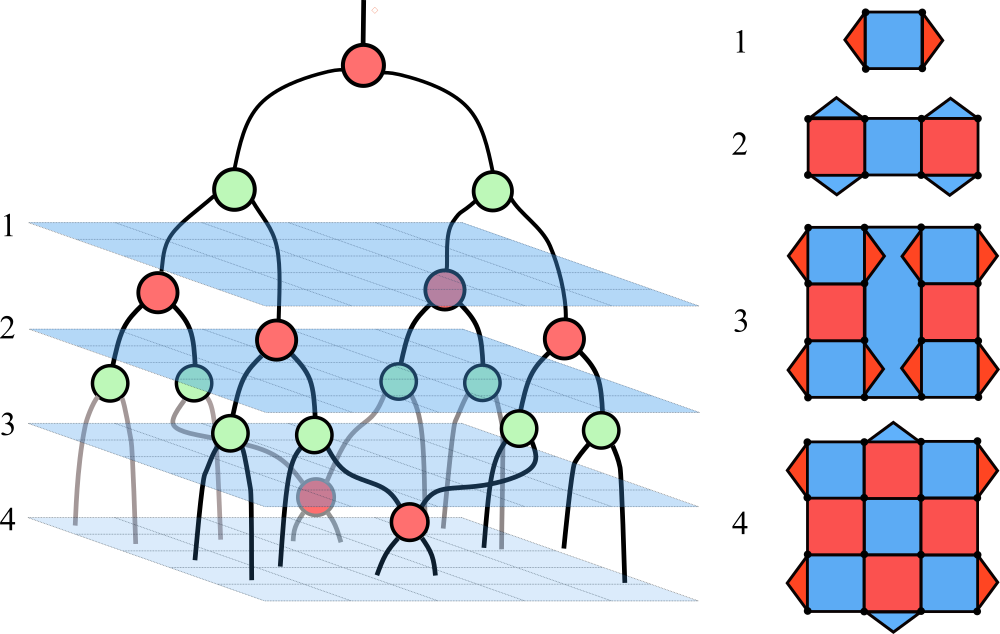}
    \caption{Layer 1 represents a $[[4,1,2]]$ code while layers 2 and 3 increases the $Z$ and $X$ distances respectively through code concatenation. Finally, non-isometries are applied at layer 4 to decrease the check weights. The corresponding stabilizer code each layer generates is marked accordingly on the right, where red or blue plaquettes denote weight-4 $X$- or $Z$-checks while triangles denote weight-2 checks.}
    \label{fig:surface_code_TN}
\end{figure}
The protocol starts from a $[[4,1,2]]$ seed code, then sequentially concatenate using $X$- and $Z$-repetition codes on the boundary of the code to increase the distance. When the check weight exceeds a prescribed limit, we reduce the weight using the non-isometric element on the boundary. An explicit sequence is shown in Figure~\ref{fig:surface_code_TN} where only 4 layers are drawn reduce clutter. However, this procedure can simply be iterated to grow out a rotated surface code to arbitrary size. 

We can understand this by examining this example in a little more depth. Let us first note the useful graphical transformation rule of the concatenation and weight reduction steps. Code concatenation to increase the $Z$ (or $X$) distance stretches the size of the red (or blue) plaquette operators. Then weight two operators of the opposite color are added between the newly grown sites and the previous boundary qubits. For instance, going from layer 1 to 2, the red side triangles get stretched to red squares while adding blue triangles between the boundary qubits of the $[[4,1,2]]$ code and the new boundary formed from concatenation. Going from layer 2 to 3, we now perform a bit-flip code concatenation that increases the $X$-distance. This procedure now stretches the blue plaquettes and triangles in the vertical direction, then adding weight 2 checks in the appropriate edges of the opposite type. Finally, non-isometric tensors of the X type are applied to the 2nd and 3rd qubits on rows 1 and 4. Recall that each $X$-type non-isometric tensor splits off an X check of weight 2 from the previous check, and then merges any $Z$-checks that only have support on one of the qubits. Therefore, they convert each pair of weight-2 $Z$-checks in the middle of layer 3 into $Z$ plaquettes in layer 4.

This procedure can then be repeated in the following sequence to increase the code distance by one without ruining sparsity: 1. increase $Z$-distance by concatenating all boundary qubits along the vertical direction using phase-flip codes. This stretches any weight-2 $Z$-check into weight-4 $Z$ plaquettes, and stretches the original weight-4 $Z$ plaquettes into weight-6 checks. It also adds weight-2 $X$-checks along the newly added edges. 2. Apply $Z$-type non-isometry to all weight-6 $Z$ checks, splitting them into weight-4 Z plaquettes and weight-2 $Z$-checks as they merge the weight-2 $X$-checks added from step 1 into weight-4 $X$ plaquettes. 3. Repeat the same steps for the boundary qubits along the horizontal direction but use bit-flip codes for concatenation to increase the $X$ distance and $X$-type non-isometric tensors to reduce weight checks. As the procedure is self-similar, it can be carried on iteratively to produce rotated surface codes of arbitrary distances by growing a tensor network out from the boundary. 

Note that this produces a surface code tensor network which is distinct from the traditional PEPS construction --- code concatenation produces the tree structure while the intermixing of the non-isometries breaks the exact tree-like expansion. It's also different from the MERA construction \cite{Aguado_2008} of Kitaev's quantum double model, as the code is grown from the boundary and makes extensive use of non-isometry during contraction. The new structure may lead to more efficient contractions, especially for weight enumerator calculations. We leave a more careful analysis and application of such tensor networks for future work.

\subsection{Example: Quantum Compass code}
The non-isometric tensors can be also be used to produce other examples, e.g. the 2d compass codes. 
The representation is certainly highly non-unique. Here we first produce a Bacon-Shor code through concatenation then deform the code to other ``gauges'' through non-isometric tensor contractions. 
Since we are working with planar geometries, for the sake of clarity and simplicity, we will only use the 2d representation and the associated graphical transformations as in the surface code example above. 

Recall that a Bacon-Shor code $[[\ell^2,1,\ell]]$ can be obtained by concatenating each physical qubit of the $\ell$-qubit phase flip repetition code with a copy of the $\ell$-qubit bit-flip repetition code. These repetition codes are contractions of $X$- and $Z$-spiders of valence 3, and so they themselves can be built up from 2-qubit repetition codes \cite[Fig 22]{QL2}. The resulting code contains weight-2 $ZZ$ generators on every vertical edge while the $X$ checks are weight $2\ell$ supported on two adjacent columns (Figure~\ref{fig:compasscode}).
\begin{figure}
    \centering
    \includegraphics[width=1\linewidth]{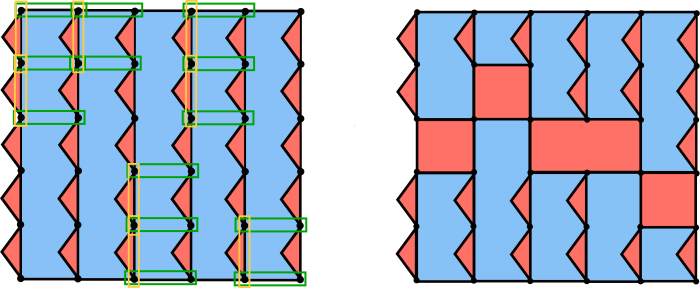}
    \caption{Left: a $[[36,1,6]]$ Bacon-Shor code where each red triangle denotes a weight-2 $Z$ check while blue rectangles denote weight-12 $X$ checks. Physical qubits with green boxes are contracted with a valence-4 $X$-spider (i.e. $X$-type non-isometric tensor) which are then followed by contracting valence-4 $Z$-spiders (i.e. $Z$-type non-isometric tensor) on qubits in yellow boxes. Right: the resulting $[[36,1,6]]$ 2d compass code with lower weight checks where $X$- and $Z$-checks are marked as blue and red boxes/triangles.}
    \label{fig:compasscode}
\end{figure}

The 2d compass codes can be thought of as different gauge-fixed versions of the 2d Bacon-Shor code. Focusing on the subspace implementation for this section, we note that $X$-checks can be broken up into smaller pieces as needed by applying $X$-type non-isometries followed by a merge using the $Z$-spider from Figure~\ref{fig:zxspider}.

For example, to carve out a weight-$2m$ $X$-stabilizer from a weight-$2n$ stabilizer from rows $j+1$ to $j+m$ and along columns $i,i+1$, one applies the $X$-type non-isometry $m$ times on the pairs of qubits $(k,i)$ and $(k,i+1)$ with $k\in\{j+1,\dots j+m\}$. It is then followed by $m-1$ pairwise $Z$-type non-isometries on qubits $(k,i)$ and $(k+1,i)$ for $k\in\{j+1,\dots j+m-1\}$. This latter action can be simplified into a single $Z$-spider of degree $2m$. The resulting $X$-check has the targeted weight with weight-2 $Z$ checks supported on each edge connecting the two adjacent qubits on the same column. An example of this for $\ell=6$ is shown in Figure~\ref{fig:compasscode} (right).

One should note that depending on the size of the region that is contracted with non-isometries--which we call an island--and the connectivity across these islands, the contractibility of the tensor network also varies. For example, in the tensor network that generates the configuration in Figure~\ref{fig:compasscode}, the non-isometries change the graph connectivity by connecting qubits in the overlapping green and yellow boxes. However, these islands are not connected, and the tensor network can still be efficiently contracted in a fashion similar to the original tree tensor network but with larger branching ratio determined by the size of the largest island. If the non-isometries are highly connected such that the islands merge into a single entity of order system size, which for instance is needed to produce the surface code, then there is no obvious exact contraction scheme that is more efficient than the usual PEPS, which scales exponentially with $\ell$.

\begin{remark}
    While it is of independent coding-theoretic interest to study how small codes can be glued together to form bigger LDPC codes, it is also useful to consider the physical processes through which these tensor contractions can be performed, for instance, for encoding and state preparation. Explicitly, the contractions tied to code concatenations are simply entangling fresh ancilla with CNOT gates while the non-isometric concatenation corresponds to a projective measurement of $XX$ or $ZZ$ operators. Both are procedures that can be implemented physically --- the former with unitary gates, and the latter with two-body Pauli measurement with post-measurement Pauli gates conditioned on the classical measurement outcomes. As quantum lego blocks, it is also clear that the tensor contraction can be mapped to Bell fusions\cite{bartolucci2021fusionbasedquantumcomputation}, where only measurement-based processes are needed.  
\end{remark}

\section{Building sparse subsystem codes}\label{sec:4}
In the above examples, we oriented the non-isometries so that we reduce the check weights but are careful so that we only overlap with minimal weight logical operators on one site, hence not reducing the distance. Nonetheless, the above procedures are still not quite satisfactory as every time we apply these tensors we merge the checks of the opposite type. If we are not careful about how or where it is applied, this can drastically increase the check weight, which severely limits the level of generality we have. Without further knowledge of the graph, it seems difficult to control the extent of such merges throughout the code building process. One obvious strategy is to restrict to special graphs with particular properties. Here we take a different tack: we simply use a tensor that reduces the weight without merging operators of the opposite type. This requires us to increase the rank of this component tensor to 5 and construct sparse subsystem codes instead of sparse stabilizer codes or quantum LDPC codes. 

Subsystem codes \cite{Poulin_2005} can be constructed from a gauge group $\mathcal{G}=\langle g_1, g_2, \dots, g_r\rangle$ with generators $\{g_i\}$. Different from conventional stabilizer codes, the group $\mathcal{G}$ need not be Abelian. 
\begin{definition}
A subsystem code is CSS-like if there exists generators of its gauge group $\mathcal{G}$ each of which is expressed solely as a tensor product of $X$ and $I$, or a tensor product of $Z$ and $I$. 
\end{definition}

In this work, we focus on CSS-like subsystem codes for simplicity.

As is usually the case, it is conceptually simpler to think of the subsystem code as a stabilizer code whose stabilizer group is the center of the gauge group $\mathcal{S}=Z(\mathcal{G})$. We then split the logical degrees of freedom of this stabilizer code into two classes: the ones that we use to encode information (the logical qubits) and the ones we sacrifice (the gauge qubits). In the Hilbert space language, the code subspace of the stabilizer code $C$ is expressed as
\begin{equation}
    C=C_{\rm logical}\otimes C_{\rm gauge}
\end{equation}
where we will only protect information in $C_{\rm logical}$. The gauge group $\mathcal{G}$ is the group of operators that correspond to the logical operations (including identity) on the gauge qubits, but acts trivially on $C_{\rm logical}$.

\subsection{Atomic components}
We will use two $[[4,1,2]]$ legos in our constructions, each of which can be built from $Z$- and $X$-spiders.
\begin{definition}
    By $ZN_{[[4,1,2]]}$ and $XN_{[[4,1,2]]}$ we denote the $Z$- and $X$-non-isometric $[[4,1,2]]$ legos, constructed as follows. The stabilizer of $ZN_{[[4,1,2]]}$ is 
    $$S_{ZN}=\langle XXXX,ZZII,IIZZ\rangle,$$
    with choices for its logical operators $\bar{Z}=ZIZI,IZIZ$ and $\bar{X}=XXII,IIXX$. For $XN_{[[4,1,2]]}$ its stabilizer is
    $$S_{XN}=\langle ZZZZ,XXII,IIXX\rangle,$$ 
    with choices for its logical operators $\bar{X}=XIXI,IXIX$ and $\bar{Z}=ZZII,IIZZ$. When viewed as non-isometric legos, we interpret legs 1, 3 as inputs and 2, 4 as outputs. The logical qubit becomes a gauge qubit (Figure~\ref{fig:ZNtensor}).
\end{definition}

\begin{figure}
    \centering
    \includegraphics[width=0.95\linewidth]{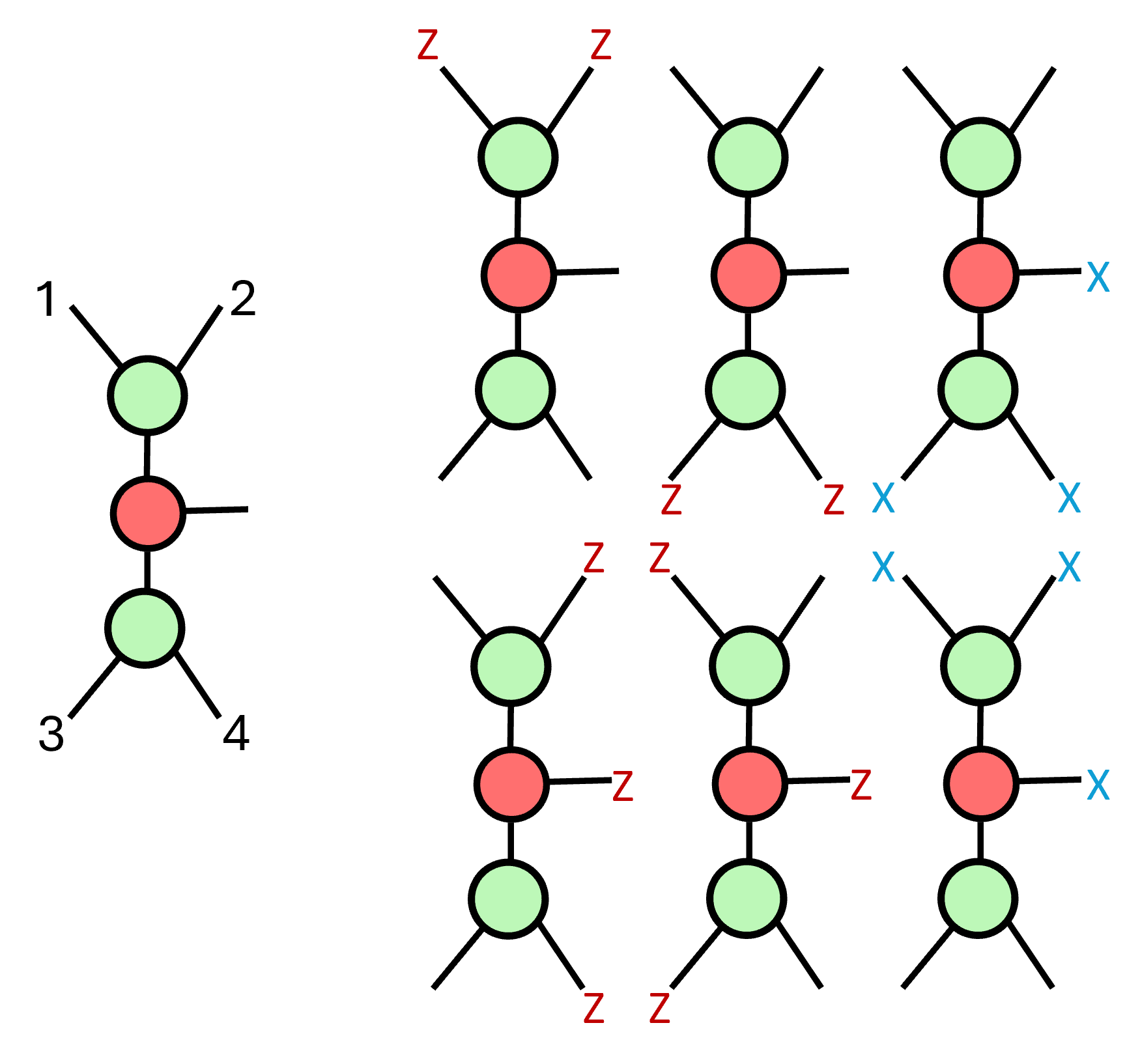}
    \caption{Symmetries of the $ZN$ tensor describing a $[[4,1,2]]$ code where the middle leg is the logical degree of freedom. $XN$ tensor is identical except exchanging the red and green colors and $X$ and $Z$ in the symmetries.}
    \label{fig:ZNtensor}
\end{figure}

Note that in our tensor representation, the logical qubit of these 4 qubit codes will never be contracted and so will become a gauge qubit in the resulting sparse subsystem codes. Hereafter, we will simply refer to this logical leg of our non-isometric legos as a gauge qubit.

Contracting the input legs (1 and 3) of the $ZN$ tensor with the output legs of an encoding map induces a map between quantum codes. This map transforms operators that have support over the two contracted legs as follows:
    \begin{itemize}
        \item any operator with support $ZZ$ on the input legs get mapped to $II$ while activating the $Z$ operator on the gauge qubit or to $ZZ$ on the output with $I$ on the gauge qubit;
        \item $ZI$ or $IZ$ operator can be mapped to $ZI$ or $IZ$ on the output without acting on the gauge qubit; 
        \item $XI$ or $IX$ operator can be mapped to $XI$ or $IX$ on the output by activating the gauge qubit with a logical $X$; and
        \item any $XX$ operator is mapped to $XX$ on the output without activating the gauge qubit, i.e. acts as identity on the gauge qubit.
    \end{itemize}
These transformations can be read off from the symmetries in Figure~\ref{fig:ZNtensor}, or recognizing that $ZZZZ$ and $XXXX$ are stabilizers of the $ZN$. For $XN$, the same results hold after swapping $Z\leftrightarrow X$ in the above.

As a consequence, the mapping $ZN$ (and similarly for $XN$) has the following effects on the gauge generators and logical operators. We only present the purely $X$ and purely $Z$ operators since any operator involving $Y$ can be produced by taking their product. 

Let $B=\mathrm{supp}(O)\cap \mathrm{supp}(ZN)$ be the non-trivial support the operator $O$ in question has with the input legs of $ZN$. We also denote bare logical operators of type $P$ as $\bar{P}_0$, gauge logical operator as $\bar{P}_g$, gauge generator of type $P$ as $g_P$ with $P=X, Z$. The transformations under $ZN$ is summarized by Table~\ref{tab:ZNtransformation}. A similar table can be compiled for $XN$ where we simply exchange the role of $X$ and $Z$.\footnote{Some examples of codes that involve non-isometric contractions of such $ZN$, $XN$ tensors can be found in \cite{QL1,evenblycode} for 2d compass/Bacon-Shor code and holographic (Evenbly) code.}

\begin{table*}[]
    \centering
    \begin{tabular}{|c|c|c|}
    \hline
        Input/Outputs & Non-active gauge qubit & Active gauge qubit  \\
        \hline
        $\bar{Z}_0$ and $|B|=2$ & no change & $\bar{Z}_g$ where $\mathrm{supp}(\bar{Z}_g)=\mathrm{supp}(\bar{Z}_0)\setminus B$\\
        \hline
        $\bar{Z}_g$ and $|B|=2$ & no change & $\bar{Z}_g'$ or $\bar{Z}_0'$\footnote{This is because gauge legs might become correlated in the tensor network.} s.t. $\mathrm{supp}(\bar{Z}_g',\bar{Z}_0')=\mathrm{supp}(\bar{Z}_g)\setminus B$ \\
        \hline
        $g_Z$ and $|B|=2$ & no change &  $g_Z'=ZZ$ on $B$ and $g_Z''$ on $\mathrm{supp}(g_Z)\setminus B$\\
        \hline
        $\bar{Z}_0$ and $|B|=1$ & no change & $\bar{Z}_g'$ with $\mathrm{supp}(\bar{Z}_0)\cup B\setminus \mathrm{supp}(\bar{Z}_0)\cap B$\\
        \hline
        $\bar{Z}_g$ and $|B|=1$ & no change & $\bar{Z}_0'$ or $\bar{Z}_g'$ with  $\mathrm{supp}(\bar{Z}_g)\cup B\setminus \mathrm{supp}(\bar{Z}_g)\cap B$\\
        \hline
        $g_Z$ and $|B|=1$ & no change & $\mathrm{supp}(g_Z)\cup A\setminus \mathrm{supp}(g_Z)\cap B$\\
        \hline
        $\bar{X}_0$, $\bar{X}_g$, $g_X$ and $|B|=2$ & no change & not permitted\\
        \hline 
        $\bar{X}_0, \bar{X}_g, g_X$ and $|B|=1$ & not permitted & $\bar{X}_g,\bar{X}_g',g_X'$ with the same support \\
        \hline
        Any operator and $|B|=0$ & no change & no change\\
        \hline
    \end{tabular}
    \caption{How $X$- and $Z$-type operators transform under the mapping $ZN$. For $XN$, one exchanges $Z$ and $X$.}
    \label{tab:ZNtransformation}
\end{table*}

\subsection{The algorithm}\label{subsec:algo}
The algorithm for growing sparse code is as follows. In particular, let the target $X$ and $Z$ generator weights be bounded by $w_X, w_Z$ and qubit degree be bounded by $q_X, q_Z$ for $X$ and $Z$ checks respectively. Each $w_P,q_P\geq 2$.
\begin{enumerate}[label*=\arabic*)]
    \item Initialize with a seed $[[n_0,k,d_0]]$ CSS-like subsystem code that satisfies the above sparsity condition. Identify a representative of a bare logical-$X$ ($\mathcal{L}_X=\{\bar{X}_0\}$) and bare logical-$Z$ operator ($\mathcal{L}_Z=\{\bar{Z}_0\}$) for each logical qubit. 
    \item For the $j$-th logical $\bar{Z}_{0}^{(j)}$ in $\mathcal{L}_Z$,  concatenate each qubit in the $\mathrm{supp}(\bar{Z}_0^{(j)})$ using a 2-qubit bit-flip code (3-valent $Z$-spider) to increase $X$ distance.
    \item Each $X$ check will get increased by even number of qubits. Pair them up in the support of each check exceeding the weight limit, apply $XN$ to each pair until the check weight is below $w_X$. 
    \item For each $q\in \mathrm{supp}(\bar{Z}_0^{(j)})$, shift one of the $Z$ checks on $q$ to its newly added partner introduced during concatenation to reduce qubit degree.   
    \item Repeat 2)-4) for all bare logical operators in $\mathcal{L}_{Z}$. 
    \item Repeat 2)-5) but with $X\leftrightarrow Z$ to increase the $Z$ distances. 
\end{enumerate}
Note that it is efficient to also track how the check operators and the bare logical operators transform under these steps because we can simply track their support during the growth process, instead of finding a new set of bare logical operators for each iteration. Hence the whole process is polynomial in $n$. 

\begin{theorem}\label{thm:1}
Steps 2-6 increase the distance of the code by at least $1$, while $k$ and the degree of the Tanner graph remains unchanged.
\end{theorem}

We provide a proof of this theorem in Appendix~\ref{app:a}. The following corollary follows immediately from the theorem, and so we omit the proof.

\begin{corollary}
    By repeating steps 2-5 and growing $X$ or $Z$ distances ($d_X,d_Z$) at different frequencies $m_X, m_Z$, one can produce asymmetric codes whose worst case distance bound satisfy $d_X/d_Z= m_X/m_Z$ in the large $n$ limit. 
\end{corollary}

Note that this process does not increase the total number of logical qubits $k$, as we have designated all subsequent logical legs to be gauge qubits. To encode enough logical qubits, we need to start with a sparse seed code with large enough $k$. As we have no requirement that $d_0$ has to be large, one can generate a sparse subsystem seed code randomly, use existing sparsification methods like \cite{hastings2016,hastings2023,Vasmer2024,He:2025exv,sparsecodecircuit}, or by applying the non-isometric codes $ZN, XN$ for sparsification.

\subsection{Check matrix transformations}

As we mentioned earlier, for each graphical move that glues together the tensor, there is a corresponding conjoining operation on the check matrix based on operator matching. 
The transformation rule for the check matrix is actually surprisingly simple. Because the code is CSS-like, we write the ``check matrix'' as 
\begin{equation}
    H= H_X\oplus H_Z
\end{equation}
where the rows of $H_X$ and $H_Z$ correspond to the X and Z ``checks'' respectively. However, the rows of $H$ need not have vanishing symplectic inner product since we are dealing with a subsystem code.

Without loss of generality, we give the descriptions for growing the $X$ distance. The protocol for growing $Z$ distance is identical with exchanging the role of $X$ and $Z$. 

The transformations consist of three key moves: concatenation, non-isometric trace, and shifting of the $Z$ checks to reduce $q_Z$. In addition, one needs to keep track of how the bare logical operators transform under these moves to keep the algorithm efficient.

\paragraph{Concatenation.}
Concatenation by repetition code is rote. Given the support of $\bar{Z}_0^{(\ell)}$ for a bare operator of the $\ell$-th logical qubit, we identify the submatrix $h_X$ generated by keeping the columns that correspond to qubits living in the $\mathrm{supp}(\bar{Z}_0^{(\ell)})$. Then the updated check matrix is $H_X'=(H_X|h_X)$ which has $|\mathrm{supp}(\bar{Z}_0^{(\ell)})|$ additional columns.

At the same time we add one additional row $v_{j,i_j}$ for each qubit in the support of $\bar{Z}_0^{(\ell)}$ to $(H_Z|\mathbf{0})$ where $\mathbf{0}$ is a zero matrix of $|\mathrm{supp}(\bar{Z}_0^{(\ell)})|$ columns. Here $v_{j,i_j}$ is a row vector of $n+|\mathrm{supp}(\bar{Z}_0^{(\ell)})|$ entries where the only nonzero entries are on the $j$-th column, which is in the support of $\bar{Z}_0^{(\ell)}$ and the corresponding qubit added in concatenation which we label as the $i_j$-th column. Counting from left to right, suppose column $j$ is the $a$-th qubit in $\bar{Z}_0^{(\ell)}$, then $i_j=n+a$. 

One must also keep track of the symplectic vectors representing the bare logical operators. For each of the bare logical $X$ operators $V_x^{(j)}$, for $j=1,\dots,k$, we again generate the subvector $v_x^{(j)}$ from the columns that have support in $\bar{Z}_0^{(\ell)}$ and append such that $V_x'^{(j)}=(V_x^{(j)}|v_x^{(j)})$.

No changes need apply to the bare logical $Z$s because their weights are unchanged. However, we set $V_z'^{(j)}=(V_z^{(j)}|\mathbf{0})$ where $\mathbf{0}$ is a length $|\mathrm{supp}(\bar{Z}_0^{(\ell)})|$ vector.

\paragraph{Non-isometric trace.}
Now we apply $XN$ to reduce the check weights. For each row in $H'_X$ we check its Hamming weight. If we find a row that is overweight, apply $XN$ to the newly added section of the check matrix $h_X$. Note that each row $v_x$ of $h_X$ must have even weight. Let $w_x$ be a vector that is only nonzero at the first two non-zero entries of $v_x$. We then subtract $w_x$ from each row of $h_x$ if both of the entries of that row are nonzero. We then modify $H'_X$ by adding the row vector $(\mathbf{0}|w_x)$ to it. We assume the vector is padded with enough zeros. Repeat this process until all rows weights of the original $H'_X$ has been checked. Define this new check matrix as $H''_X$. 

Since the non-isometric trace only converts weight $2$ $Z$-type stabilizers to gauge operators, no changes need apply to $H'_Z$. We simply set $H''_Z=H'_Z$ for this stage. For bare logicals, recall that $XN$ does not activate its gauge degree of freedom if $X$-type Pauli strings are pushed to identical Pauli strings (i.e. $XX$ is pushed to $XX$, $IX$ to $IX$ and so on), and so the symplectic vectors of the bare logicals $V_x'^{(j)}$ are unchanged by this process. The bare $Z$ logicals are also unchanged as $XN$ does not act on their support. The new check matrix is $H=H_X''\oplus H_Z''$.

\paragraph{Shifting checks.} 
Finally we use the newly added weight-$2$ $Z$ generators to move some of the $Z$ checks off of data nodes that exceeded $q_Z$ by $1$. For each column in $\mathrm{supp}(\bar{Z}_0^{(\ell)})$, check the column weight. If the Hamming weight of column $j$ exceeds $q_Z$, add $v_{j,i_j}$ to it. Repeat for all such columns that are overweight. 

\paragraph{Iterate.} 
For each $\ell=1,\dots,k$, repeat the above process so that all $d_X$ is increased by at least one for all the logical qubits. To increase $d_Z$, repeat the above and exchange the role of $X$ and $Z$.

\subsection{Asymptotic Scaling}

\begin{theorem}
    The worst case asymptotical scaling of this algorithm will produce sparse codes with parameter $kd^2=O(n)$. 
\end{theorem}

The proof is given in Appendix~\ref{app:b}. We hasten to point out that this scaling saturates the Bravyi-Poulin-Terhal bound on 2d codes \cite{Bravyi_2010}, and is achieved on codes like the 2d Bacon-Shor code. Hence this theorem is optimal given the level of generality. 

However, in generality, it is not clear what the average case scaling should be. Ultimately, this type of scaling depends on the minimal set that one can concatenate to increase the code distance by 1. In other words, what is the minimal set $S$ such that $S$ intersects with all minimal weight logical operator, such that applying concatenation to it will increase the minimal distance? In our algorithm we have chosen a set $S'=\mathrm{supp}(\bar{Z}_0)$ or $\mathrm{supp}(\bar{X}_0)$ such that concatenation adds to all operator weight in the system. In general this need not be the case, and the structure of this set can depend on the type of code or tensor network in consideration. It remains to be seen how to identify $S$. Furthermore, the conjoining step that performs the non-isometric trace does not always reduce to the worst case for all logical qubits. This can also significantly improve the distance scaling.

\subsection{Examples}
Let us look at a few simple examples to build some intuition.

\subsubsection{2d Bacon-Shor code}
Suppose we are given as an initial seed the $2\times 2$ Bacon-Shor code whose $X$ and $Z$ gauge generators are shown in red and blue (Figure~\ref{fig:2dbsc}). Its bare logical $\bar{Z}_0$ is supported on qubits 1 and 2 while the bare $\bar{X}_0$ logical on qubits 2 and 4. Following the algorithm, we concatenate by $2$-qubit phase-flip codes on $\mathrm{supp}(\bar{X}_0)$ (orange circle) to increase the $Z$ distance. Going from (i) to (ii), $d_Z$ is increased by $1$, but so is the generator (red plaquette). One can then reduce the check weight using a $ZN$ on qubits in the yellow circle. This produces two new gauge generators in (iii). Then one repeats by concatenating the $\mathrm{supp}(\bar{Z}_0)$ (purple circle) with bit-flip codes to increase $d_X$ by $1$. The resulting code again has increased the check weights above our threshold, hence we apply $XN$ to each check to reduce their weights back down to two (iv). Finally we arrive at (v), which is a $3\times 3$ Bacon-Shor code. The same algorithm can then be run iteratively to increase distance to arbitrary $d$. 

\begin{figure}
    \centering
    \includegraphics[width=1\linewidth]{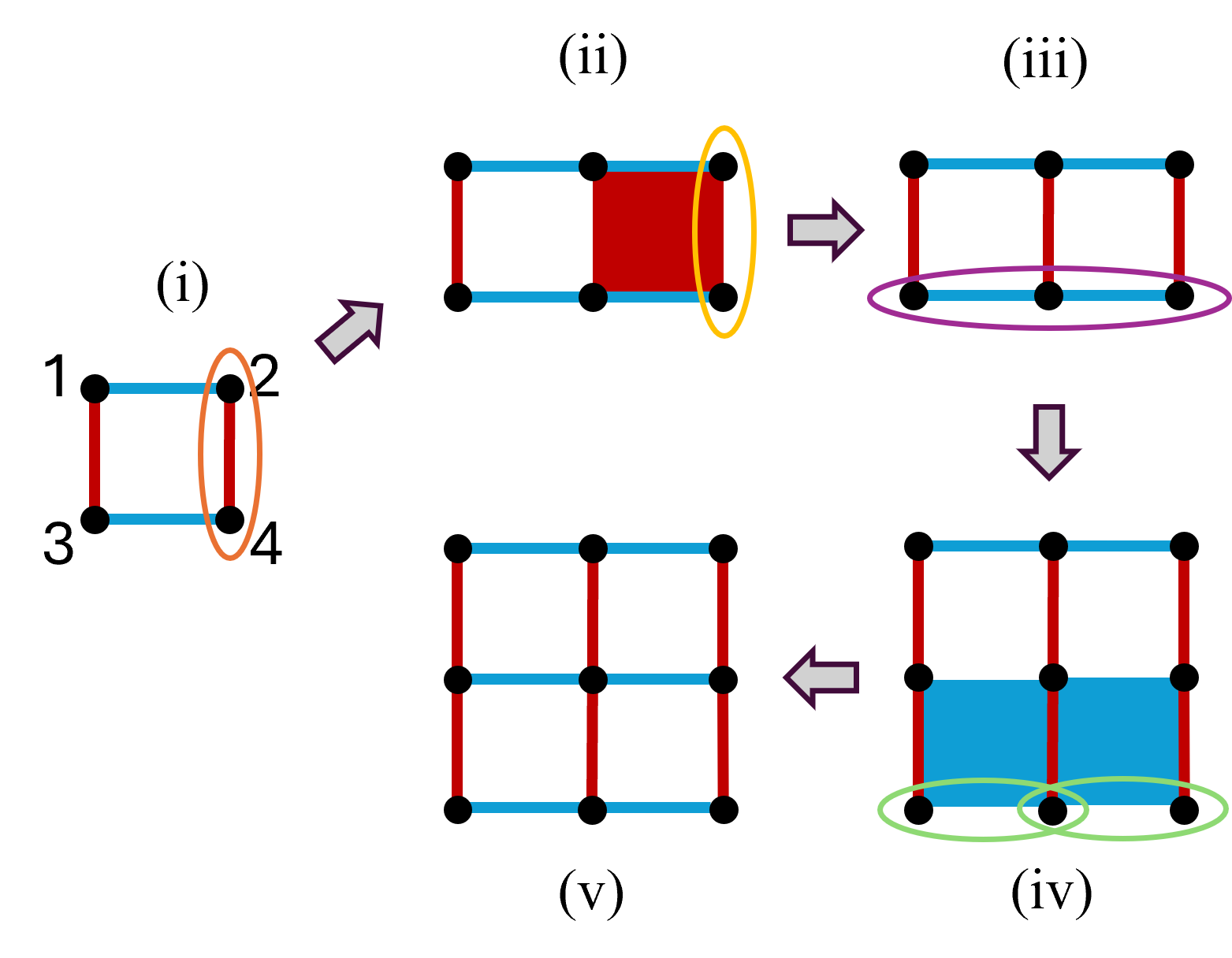}
    \caption{How a 2d Bacon-Shor code can be grown using the algorithm by keeping the target weight as $w_x=w_z=2.$}
    \label{fig:2dbsc}
\end{figure}

\subsubsection{3d Bacon-Shor Code}
Other codes can also be generated by changing the initial seed. For instance, given an $8$-qubit 3d Bacon-Shor code, the algorithm can iteratively build a 3d Bacon-Shor code, Figure~\ref{fig:3dbsc}. As before, the $X$- and $Z$-type gauge generators are weight-2 and labeled as colored edges in blue and red respectively.

Here the bare logical $\bar{X}_0$ is supported on a face, e.g. the qubits circled in orange in (i). A bare logical $\bar{Z}_0$ is supported on a blue edge. Following the algorithm, we can increase the $Z$ distance by concatenating $\mathrm{supp}(\bar{X}_0)$ with phase-flip codes (circled in orange). This also increases the weight of the gauge generators, red faces in (ii). Then we reduce the weight of each face by acting $ZN$ on pairs of qubits circled in yellow, which produces (iii). Then doing the opposite, we concatenate $\mathrm{supp}(\bar{Z}_0)$ (circled in purple) with bit-flip codes to increase $X$ distances. The resulting code (iv) now has weight-$4$ $X$-type gauge generators, which we split and weight reduce by acting $XN$ on the qubit pairs circled in green. Finally in (v), this produces a 3d Bacon-Shor code with one higher distance without ruining the sparsity. 

Repeating the protocol that concatenates the $\mathrm{supp}(\bar{X}_0)$ of (v) and weight reduce, we arrive at (vi). Doing the same to $\mathrm{supp}(\bar{Z}_0)$ we arrive at (vii). Finally, one can choose alternative set of the gauge operators through multiplication. Note that a few red edges, which correspond to elements of the gauge group are also added to make manifest the symmetry of the code. Nevertheless, as generators, some of the colored edges in (viii) are redundant in 3d Bacon-Shor codes. 
\begin{figure*}
    \centering
    \includegraphics[width=0.9\linewidth]{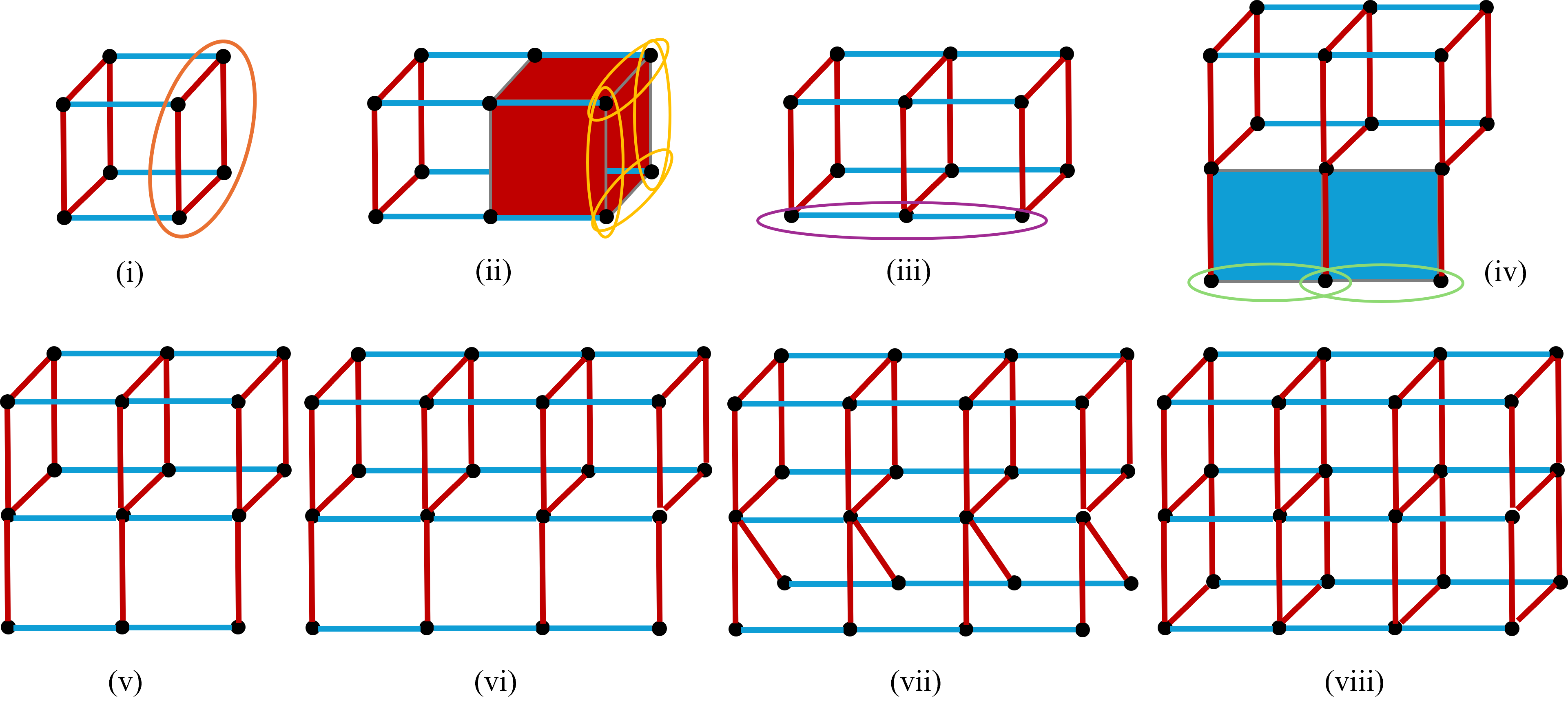}
    \caption{A sequence of concatenation and non-isometric concatenations following the algorithm is used to grow the 3d Bacon-Shor code going from (i) to (viii). The detailed intermediate steps are omitted in the figure going from (v) to (viii). }
    \label{fig:3dbsc}
\end{figure*}

The distance lower bound suggests that the code at (viii) has $d\geq 4$. While $d_Z=4$ saturates this lower bound, $d_X=6$ can be larger, which is different from the 2d example\footnote{Note that the code does not violate the asymptotic distance lower bound as one would naively expect $d\sim n^{1/3}$ in such 3d codes. This is because the code we build does not have equal length in all three dimensions.}. Additionally, the apparent spatial arrangement on a regular cubic lattice is optional, as we can grow out the code along, e.g. $\mathrm{supp}(\bar{Z}_0)$ in other directions and the apparent lattice structure depends entirely on how the red edges are shifted through operator multiplication. Here we use this simple arrangement for the sake of clarity. 

\subsubsection{More general Tanner graphs}
As the algorithm allows for arbitrary weight and degrees, we need not keep the operator weight to 2. The algorithm extends to more general codes, where the Tanner graph has higher but bounded degree. However, these codes usually are not locally embeddable on a manifold with low spatial dimension.

As a demonstration, we grow a sparse code from a $[[4,2,2]]$ seed code with limits such that $q_X\leq 5, q_Z\leq 4, w_X\leq 4$ and $w_Z\leq 4$. Automating the transformation rules on check matrices, we can easily obtain the Tanner graph associated with this subsystem code (Figure~\ref{fig:gen_Tanner}). For the sake of clarity, we build a small $[[684,2,12\leq d\leq 32]]$ code. The lower bound on distance is given by Theorem~\ref{thm:1} while the upper bound is given by tracking the bare logical operators. A tighter limit can probably be obtained through row operations on the check matrix or through the tensor enumerators, though this would require further work for efficient adaptation to subsystem codes \cite[\S{III.D-E}]{QL2}.

\begin{figure}[htp]

\begin{subfigure}{\linewidth}
\includegraphics[width=0.9\textwidth]{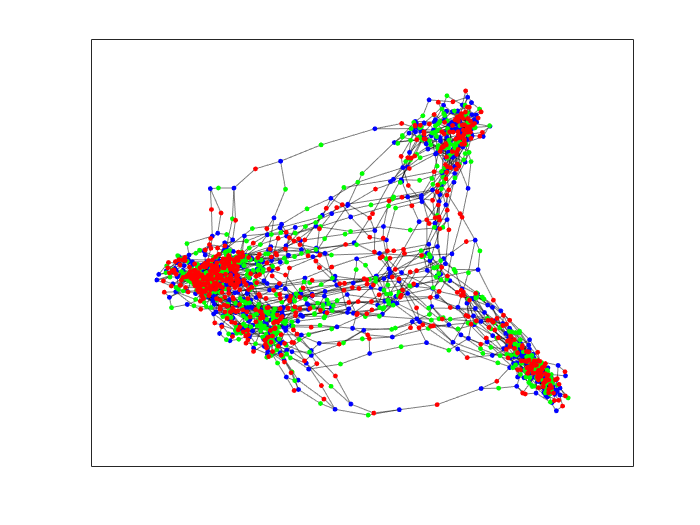}
\caption{Red:Z check nodes, Green: X check nodes, Blue: data nodes.}
\end{subfigure}

\bigskip

\begin{subfigure}{\linewidth}
\includegraphics[width=0.9\textwidth]{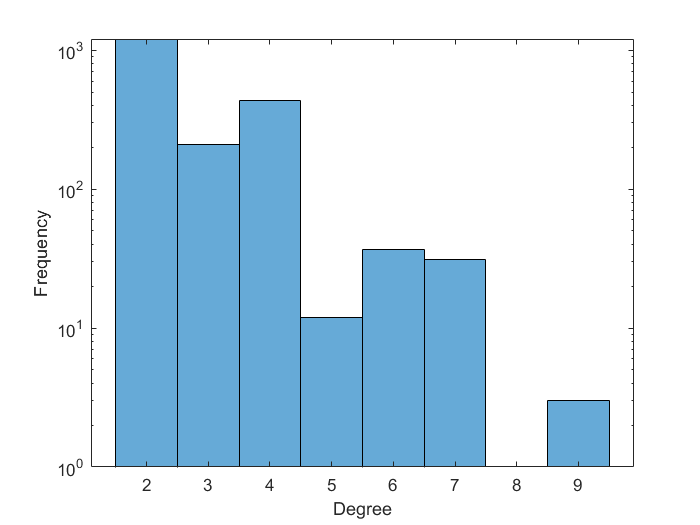}
\caption{Degree distribution of all nodes in the Tanner graph.}
\end{subfigure}

\caption{Tanner graph and degree distribution of a $[[684,2,12\leq d\leq 32]]$ code generated by the iterative algorithm. We have set the upper limits $q_x\leq 5, q_z\leq 4; w_x\leq 4, wz\leq 4$. }
\label{fig:gen_Tanner}
\end{figure}

\section{Is repetition code enough?} \label{sec:5}

A concern from the previous sections may be that the use of only repetition codes would severely limit the variety of codes we can construct. For example, one may worry whether these lego blocks can only obtain examples that are limited variations of the surface code or generalized Bacon-Shor codes. A part of this worry is alleviated by the generic algorithm in \S{\ref{subsec:algo}}, where a wide variety of subsystem codes may be constructed. The same is far less clear for qLDPC codes based only on the descriptions in \S{\ref{sec:qldpc}}. 
It is also unclear that the lego blocks and the algorithm we have will be enough to cover good sparse quantum codes that have linear rate and distance. Fortunately, this worry is misplaced: all CSS codes can be built from the basic components we use here with tensor contraction, i.e., conjoining. 

\begin{theorem}
    Every CSS codes can be built from $2$-qubit bit-flip and phase-flip repetition codes through conjoining.
\end{theorem}

\begin{proof}
We will prove this graphically using tensor contractions by simplifying a measurement-based state preparation protocol. There is a similar proof for the ZX calculus  \cite{kissinger2022phasefreezxdiagramscss}, which was brought to our attention as this manuscript was being completed. Nevertheless, it can be instructive to briefly walk through the proof from a physically intuitive perspective below \cite{QL2} . 

Suppose we are given the $X$- and $Z$-checks of a CSS code and consider any measurement-based state preparation protocol to prepare a codeword, for instance by measuring the checks and post-selecting the trivial syndrome. Physically, one can perform such measurement-based state preparation process by introducing an ancillary qubit for each check, entangling the physical qubits with the ancilla, then measuring the ancilla in the computational basis and post-selecting on the desired outcome. As the circuit elements themselves are tensors and can be contracted and simplified, the resulting tensor network can be reduced to a Tanner graph. An example is shown in Figure~\ref{fig:ZXCSSa} for $Z$-checks\footnote{See also \cite[\S{E.4}]{QL2} for detailed derivations for more general cases where the stabilizer group can be non-Abelian.}.
\begin{figure*}
    \centering
\includegraphics[width=0.7\linewidth]{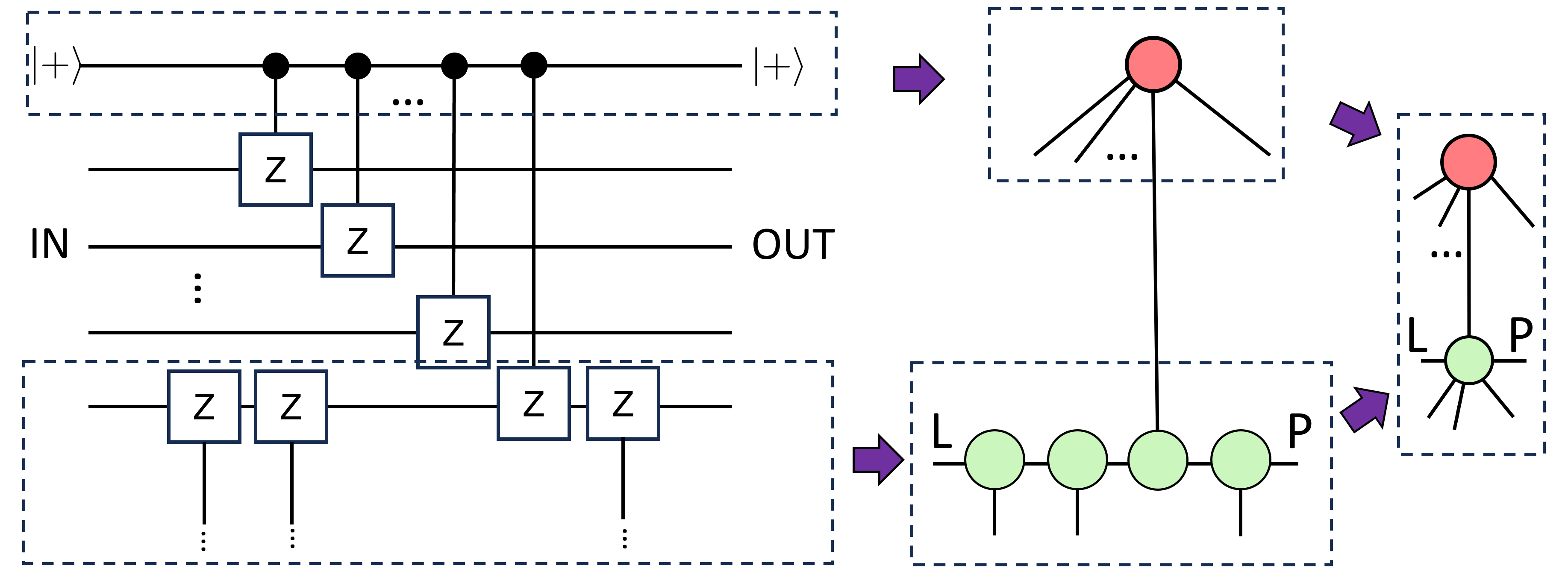}
    \caption{Entangling an ancillary qubit and then postselecting for the trivial syndrome on the $Z$-checks produce a unphased X-spider for the check node and an unphased Z-spider for the data node. L, P here denote logical/input and physical/output legs respectively.}
    \label{fig:ZXCSSa}
\end{figure*}

Without loss of generality, we will perform the $X$-checks after the $Z$-checks. The derivation is similar to Figure~\ref{fig:ZXCSSa} but exchanges the colors. Since the circuit for the $X$-check follows that of the $Z$, the resulting tensor network of the code can be read off as Figure~\ref{fig:ZXCSSb}. Since any $X$- or $Z$-spider can be built from contracting the tensors of two-qubit phase and bit-flip codes, this completes the proof.
\end{proof}

\begin{figure}
    \centering
\includegraphics[width=\linewidth]{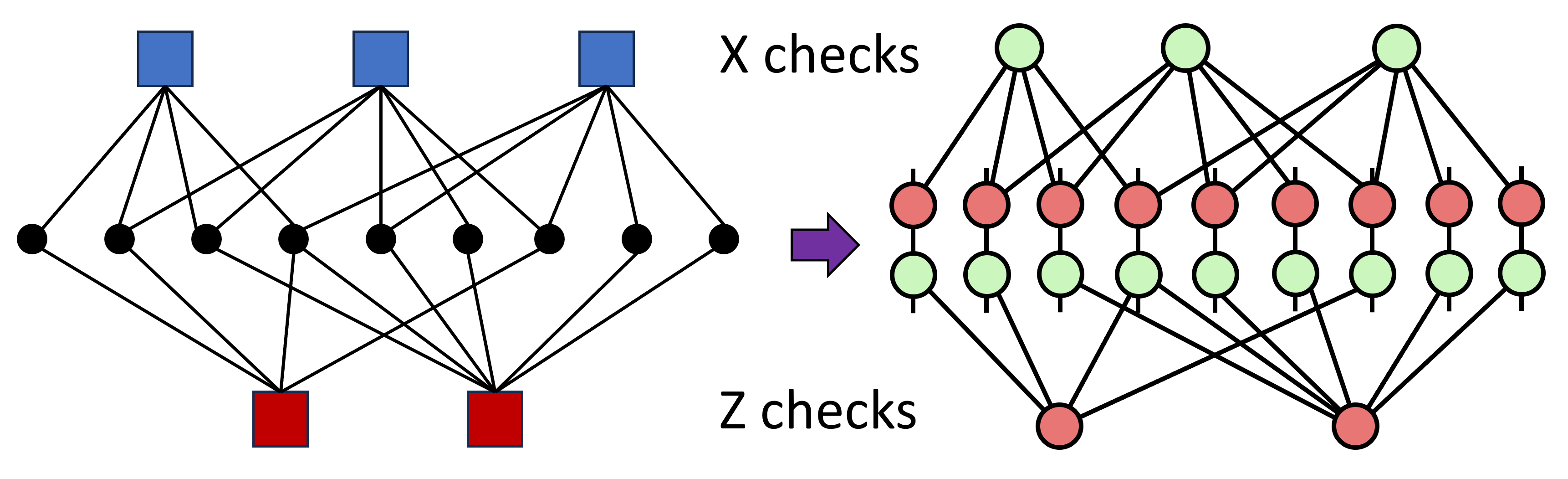}
    \caption{A CSS code can be built as a tensor network of bit-flip and phase-flip codes with an identical the Tanner graph. Here blue and red nodes label X and Z checks of the CSS code respectively.}
    \label{fig:ZXCSSb}
\end{figure}

Note that tensor networks built this way \cite{QL2} have a large kernel as not all input logical legs represent independent logical qubits, similar to \cite{Bombin_2024,ZXcalc}. Such a representation has been useful for fusion based quantum computation and for understanding fault-tolerant complexes \cite{bombin2023faulttolerantcomplexes}. 

It is also straightforward to create a network that mods out the kernel in the case where the number of input legs equals the number of logical qubits. Let $S$ be the stabilizer group that defines an $[[n,k]]$ CSS code. The Choi state of its encoding isometry is a stabilizer state over $n+k$ qubits and is stabilized by the group of elements generated by $\bar{L}_n\otimes L_k$ where $L$ is the tensor product of $X$ and identity (or $Z$ and identity) supported on the last $k$ qubits.  $\bar{L}$ is any physical representation of $L$ over $n$ qubits in the original stabilizer code. 

Since this state can be written such that its stabilizer has only $X$- and $Z$-generators, it can be prepared using our prescription above where the resulting tensor network has only physical legs. Then we perform the usual code shortening by going back to the encoding map of the code stabilized by $S$ using the channel state duality. To do this, we simply redesignate the last $k$ legs as logical inputs and the first $n$ legs as physical outputs. This latter  construction is identical to the one by \cite{kissinger2022phasefreezxdiagramscss}. 

Finally, it is clear that repetition codes can build more than CSS codes: we have already seen this by constructing CSS-like sparse subsystem codes using the algorithm of\S{\ref{subsec:algo}}. Therefore, these repetition code building blocks are sufficient for building up non-trivial sparse quantum codes, we just have to be creative enough in combining them in an efficient manner. 

\section{Discussion}\label{sec:discussion}

In this work, we answer the question whether it is possible to grow sparse quantum codes in ways similar to concatenation. We showed that some observations from \cite{Farrelly_2022,QL1} can be generalized to build sparse codes far more systematically by recognizing the power of concatenating or conjoining with ``bad codes'' that have low distances for high logical weight operators to reduce stabilizer check weights. We provided an efficient algorithm with a lower bound guarantee on the distance for building CSS-like subsystem codes. This simple algorithm can produce subsystem codes that have asymmetric distances and only require two types of lego blocks: bit-flip and phase-flip codes. Using these components, we produced new tensor network constructions for the 2d surface code, 2d compass code, 2d Bacon-Shor code and 3d Bacon-Shor code. More generally, these simple components suffice to produce asymptotically good codes as they can produce any CSS code. For nearer terms, the protocol to grow a LDPC or subsystem code can be particularly beneficial when applied to codes that maintain some form of spatial locality, for example bivariate bicycle codes \cite{Bravyi_2024}, where a family of subsystem codes of various sizes can be used to interpolate between the known construction of various dimensions. We leave the construction of such an algorithm to future work. 

We emphasize that although there are already many methods to generate quantum LDPC codes with better asymptotic parameters, a key purpose of our work here is to develop a better understanding of qLDPC codes from a ``reductionist coding-theoretic'' perspective --- we ask how sparse codes can emerge from smaller atomic codes. This is similar to Tanner codes in spirit, but distinct in construction. It is also a step to understand sparse codes from the point of view of TN/QL and operator flow.  

There are many directions that requires further exploration. This involves a more general algorithm for building qLDPC codes where the checks commute as well as identifying more specific methods and graph structures that allows us to improve the scaling of distance and rate during the iteration. In particular, it suggests that the study of minimal support sets that intersect with all minimal weight logical operators can provide valuable constraint for this method. On the computational front, more efficient methods and codes will be needed to determine the average distance scaling of our algorithm. This will also involve the development of software packages that can automate the computation of (higher genus) weight enumerators for these codes using tensor networks.

The algorithm of \S{\ref{subsec:algo}} can also be modified to improve its rate. While we have added all the logical qubits at the beginning, we need not do this in general. For instance, one can gradually add to the code additional logical qubits using the generalized isometric concatenation used in holographic codes, which introduces new logical degrees of freedom as the code and distance grows. Alternatively one can add logical qubits later during the concatenation, which have lower distance. This opens up an alternative route to construct holographic-like codes that also remain sparse at the ``boundary'' of the tensor network.

The new method of concatenation/conjoining with non-isometric codes or codes with poor distances also opens up several exciting future directions such as constructing sparse codes with (targeted) fault-tolerant gate sets and growing sparse codes that reduces spatial non-locality, for instance by using low distance but high rate triorthogonal codes \cite{BH}. It is a major challenge to design codes that have targeted gates that act within a single code block. Coupled with quantum lego's power in this direction \cite{cao_lackey_targeted}, as well as its power to generate codes with transversal non-Clifford gates, this method can help produce novel codes that are important for universal fault-tolerant logic and magic state distillation. The spatially local embedding of sparse codes is also an important direction experimentally. Although \cite{Bravyi_2010,Krishna} showed that non-locality necessarily exist when a particular distance or rate is needed, it is possible that an algorithm that relies on local growth can significantly simplify or reduce the cost associated with non-spatially-local operations \cite{li2024transformarbitrarygoodquantum}. It is also important to examine variations of this algorithm in growing codes that have more regularity in its connectivities, which is often desirable in practical applications. Furthermore, since subsystem codes are easily constructable using quantum lego, another immediate direction is to consider atomic legos with greater varieties that can lead to subsystem codes that permit fault-tolearnt code-switching, which have thus far been mostly limited to instances of 3D color codes. As tracing small atomic codes is equivalent to the fusion of photonic resource states, adaptions of the current protocol might facilitate the FT implementation of more general sparse codes in fusion-based quantum computation\cite{Litinski:2025irj}. 

In light of the recent developments in combining machine learning and quantum code design \cite{QLRL,mauron2023optimizationtensornetworkcodes,olle2024simultaneousdiscoveryquantumerror,He:2025exv}, our method for sparse code growth also provide a much needed constraint or guideline for generating qLDPC codes using these methods. For example, the conjoining by non-isometric codes can be added as a separate move in RL-based method above to maintain sparsity in generating LDPC codes. One can also use the algorithm as is for different seed codes and characterize the various codes that can be produced. In particular, we need to understand the average case scaling with regularity constraints. It will be interesting to also pin down the constant scaling factor, which is of practical interest \cite{Liang:2025qqr}.

Our construction also presents an alternative strategy that sparsifies a code. Existing methods can sparsify pre-built codes, such as those of \cite{gottesman2022,floquetify, Bacon_2006,sparsecodecircuit,hastings2016,hastings2023,Vasmer2024}. To an extent, the iterative algorithm alternates between concatenation and sparification, the former increases code distance while the latter may be achieved using other methods above. For example, \cite{hastings2016,hastings2023,Vasmer2024,He:2025exv} show that CSS codes can be sparsified by paying a little cost of reducing distance in the asymptotic limit while \cite{sparsecodecircuit} can convert an existing fault-tolerant encoding circuit of a code into a subsystem code that is sparse. From a very different perspective, one can also reduce check weights via ``floquetification'' where dynamical measurements of non-commuting checks can be used to protect quantum information in a way different from a traditional static code subspace. It is not yet clear how the performance these various method compare with each other in a finite near-term regime as many guarantees are asymptotic. For dynamical codes, it is also challenging to characterize its FT properties. Therefore, it can be beneficial to explore whether combinations of such methods can lead to better design of sparse codes in the future.

\section*{Acknowledgement}
We thank Giuseppe Cotardo, Balint Pato, Hengyun (Harry) Zhou for comments and discussions. 
\appendix

\bibliography{ref}
\bibliographystyle{unsrt}
\appendix

\section{Proof of Theorem 1}\label{app:a}
To prove the theorem, we first set up a few lemmas. We prove the results for increasing $X$-distance; the proof for increasing $Z$-distance is analogous, by exchanging $X$ and $Z$ in all arguments. 

It is obvious that code concatenation can increase $X$-distance. Formally, we prove the following.

\begin{lemma}
Concatenation by bit-flip repetition code on $\mathrm{supp}(\bar{Z}^{(j)}_0)$ increases the weight of all $X$-type logical operators that act non-trivially on the $j$-th logical qubit by at least $1$.
\end{lemma}
\begin{proof}
    Note that all $X$-type (dressed) logical operators that have support on the $j$-th logical qubit must anticommute with $\bar{Z}^{(j)}_0$, since it is a bare logical operator. This means that each of these $X$-type operators must intersect with $\mathrm{supp}(\bar{Z}^{(j)}_0)$ at least $1$ site. As this site is expanded into two by concatenating with bit-flip codes, we add one to the $X$-weight for all such $X$-operators\footnote{For example, suppose $k=3$ and $\bar{Z}^{(1)}_0=\bar{Z}\bar{I}\bar{I}$ acts non-trivially on the first logical qubit (and trivially on all other logical and gauge qubits), then $\bar{X}\bar{X}^{a}\bar{X}^b$ all have their distances increased for $a,b=0,1$.}. It is possible that some operators have their weights increased by more than $1$ if their support overlap with $\mathrm{supp}(\bar{Z}^{(j)}_0)$ at more than $1$ site.
\end{proof}

Note that the $X$-type logical operators for other logical qubits may not have their weights increased as they commute with $\bar{Z}_0$ they might not overlap with the support of $\bar{Z}_0$. Hence, the distances of the other logical qubits need not increase.

Additionally, concatenation will increase the check weights of all $g_X$ that intersect with $\bar{Z}^{(j)}_0$. It also introduces new weight-$2$ $Z$-checks between any $q\in \mathrm{supp}(\bar{Z}^{(j)}_0)$ and its newly added partner $q'$. This adds one to the $Z$-qubit degree at each $q$ but does not violate the $Z$-weight limit $w_Z$. Nor does it alter any $Z$ distances --- it only adds equivalent representations.

\begin{table*}[]
    \centering
    \begin{tabular}{|c|c|c|c|}
    \hline
        initial & after concatenation & after XN & after shift\\
        \hline
        $w_X$ &  $w_X'\leq w_X+|\mathrm{supp}(g_X)\cap \mathrm{supp}(\bar{Z}_0^{(j)})|$ & $w_X''\leq w_X $& unchanged\\
        \hline
        $q_X$ & $q_X'\leq q_X$ & $q_X''\leq q_X$ & unchanged\\
        \hline
        $d_X^{(\ell)}$ & $d_X'^{(\ell=j)}\geq d+1, d_X'^{(\ell\ne j)}\geq d$ & unchanged & unchanged\\
        \hline
        $w_Z$ & $w_Z'\leq w_Z$ & unchanged & unchanged\\
        \hline
        $q_Z$ & $q_Z'\leq q_Z+1$ & $q''_Z\leq q_Z+1$ & $q''_Z\leq q_Z$\\
        \hline
        $d_Z^{(\ell)}$ & unchanged & unchanged & unchanged\\
        \hline
    \end{tabular}
    \caption{How the various weights are transformed by steps 2 through 4 of the algorithm for the $j$-th logical qubit. $w_X,w_Z,q_X,q_Z$ denote the maximum degrees of various nodes in the Tanner graph defined in the algorithm while $d_{P}^{(\ell)}$ where $P=X$ or $Z$ is the X or Z distance of the $\ell$th logical qubit. If the bound is unchanged compare to the previous step, then it is denoted as ``unchanged'' in the entry.}
    \label{tab:transformed_wts}
\end{table*}

How the weights change throughout the algorithm is summarized in Table~\ref{tab:transformed_wts}. 

\begin{lemma}\label{lemma:evenoverlap}
    Let $g_X$ be any $X$-type gauge generator. Then $|\mathrm{supp}(g_X)\cap \mathrm{supp}(\bar{Z}_0^{(j)})| = 0 \pmod{2}$, \textit{i.e.} they overlap at even number of sites.
\end{lemma}
\begin{proof}
    Each gauge generator $g_X$ must commute with the bare logical $\bar{Z}_0^{(j)}$, as they only act non-trivially on gauge qubits, not logical qubits. This requires that these operators to have even overlapping support.
\end{proof}
In fact, any logical $X$ operator that acts on other logical qubits must also have overlapping support with $\bar{Z}_0^{(j)}$ being even.

Now, the critical step for this procedure is to identify for each gauge generator pairs of qubits that were added during concatenation. If $g_X$ did not overlap with $\bar{Z}_0^{(j)}$ we leave it alone.

\begin{lemma}
    If for each $g_X$ whose weight get increased by nonzero amount, applying $XN$ to enough pairs for each $g_X$ produce $X$-checks that are below the sparsity limit. The resulting $X$-distance for this logical qubit is at least $d+1$ after the procedure. The procedure does not alter the $Z$-distance or $Z$-check weights.
\end{lemma}
\begin{proof}
    By Lemma \ref{lemma:evenoverlap} for each $g_X$ we can identify pairs of qubits that have just been added through concatenation. Namely, because the overlapping support of $g_X$ and $\bar{Z}_0^{(j)}$ is even, for each qubit $q$ added in the support we have also added a partner $q'$. Applying $XN$ to distinct pairs of added qubits we can clearly  reduce their check weights back to the original value while introducing new $X$-type gauge generators of weight $2$. These new $X$-gauge generators only have support over the newly added qubits and therefore do not add to the $X$-degree in the original code. As they are weight 2, they have bounded weight. 
    
    Since the original Tanner graph has bounded degree, each qubit $q$ is checked by a constant number of $X$- and $Z$-checks ($q_X,q_Z$). This means that for any newly added qubit $q'$ that came out of concatenating this qubit, both will be checked by $q_X$ number of $X$-checks after concatenation. An $XN$ reduction step above will split $q,q'$ such that they will now be checked by two types of checks: $q$ by the checks of the original code and $q'$ by the newly added weight-$2$ $X$-checks. However, since $q'$ is contained in at most $q_X$ distinct $g_X$'s, there will be at most $q_X$ weight-2 $X$-checks acting on it. Hence its $X$-degree is bounded by the same constant for sparsity. Another way of saying this is that applying $XN$ in the procedure does not exceed the $X$-type qubit degree for each newly added qubit (Figure~\ref{fig:wt_reduce_pf}).

    Let us examine how this alters the distance on the logical qubit of interest. All logical $X$ operators had their weights increased by an odd number, i.e. the number of intersections with $\bar{Z}_0^{(j)}$. The $XN$ weight reduction will introduce weight-$2$ gauge operators. The worst that can happen is that there is a gauge operator (not generator) for each pair of newly added qubits during concatenation. Then when one removes these, we can pair up these qubits except one. Hence all logical $X$-operators here will have their distances increased by at least $1$ after this because the number of newly added qubits is odd. For example, a logical $X$-operator supported on the qubits circled in orange in Figure~\ref{fig:wt_reduce_pf}, concatenation increased its weight by 3, but $XN$ will decrease it. However, because $|\{q'\}|$ is odd, there is no perfect pairing and hence no gauge operator multiplication that can completely reduce the support of the new logical operator to what it was before. 

    For other logical $X$-operators that do not act on this logical qubit, their weights all increase by an even number, so the worst case is that the newly added weights are cancelled through gauge operator multiplication and their distance do not increase. However, their distance also cannot be less than $d$ because $XN$ only reduce the weight by acting on the newly added qubits. 

    It remains to check that the non-trivial logical operators did not get converted into pure gauge operators. This can happen when the logical leg is correlated with the gauge leg we add in $XN$, i.e., any logical operator pushed through the tensor network should not always activate the gauge qubit on $XN$. Fortunately, this does not happen as any overlap can be passed with or without activating the gauge qubit, see table~\ref{tab:ZNtransformation} (but exchange $X$ and $Z$). Note that both options of active and non-active gauge qubits are permitted for $\bar{X}$. This decouples the gauge leg/qubit with other existing logical or gauge legs in the tensor network. 

    For $Z$-operators the situation is slightly different. Again going to this Table, we note that certain modes of operator pushing are not permitted without an active gauge qubit. In particular, two $Z$s are passed to two $Z$s without activating the gauge qubit whereas any single $Z$ pushing must activate the gauge qubit. This means that any $Z$-type operator with $|A|=1$ are automatically correlated with the gauge degree of freedom of the $XN$ tensor. However, since $XN$ acts only on the qubits that are newly added, the only  such operators are the weight-2 $ZZ$ stabilizer checks and any operator that one can obtain via multiplication by such stabilizers. Therefore this process converts $ZZ$ stabilizer checks into weight-2 $ZZ$ gauge checks. But because these stabilizers commute with all gauge and logical operators, they are necessarily decorrelated with the existing gauge and logical operators supported on the original $n$ qubits, and their conversion to gauge operators does not impact the $Z$-distance. Therefore, repeated applications of $XN$ can correlate the gauge qubits on different $XN$s (particularly in the case where a single newly added qubit is used in different checks, then the $Z$-gauge qubits from $q$ different $XN$s will be correlated). However, it will not correlate a logical degree of freedom from the original tensor network with the newly added gauge qubits on $XN$s.
\end{proof}

\begin{figure}
    \centering
    \includegraphics[width=0.95\linewidth]{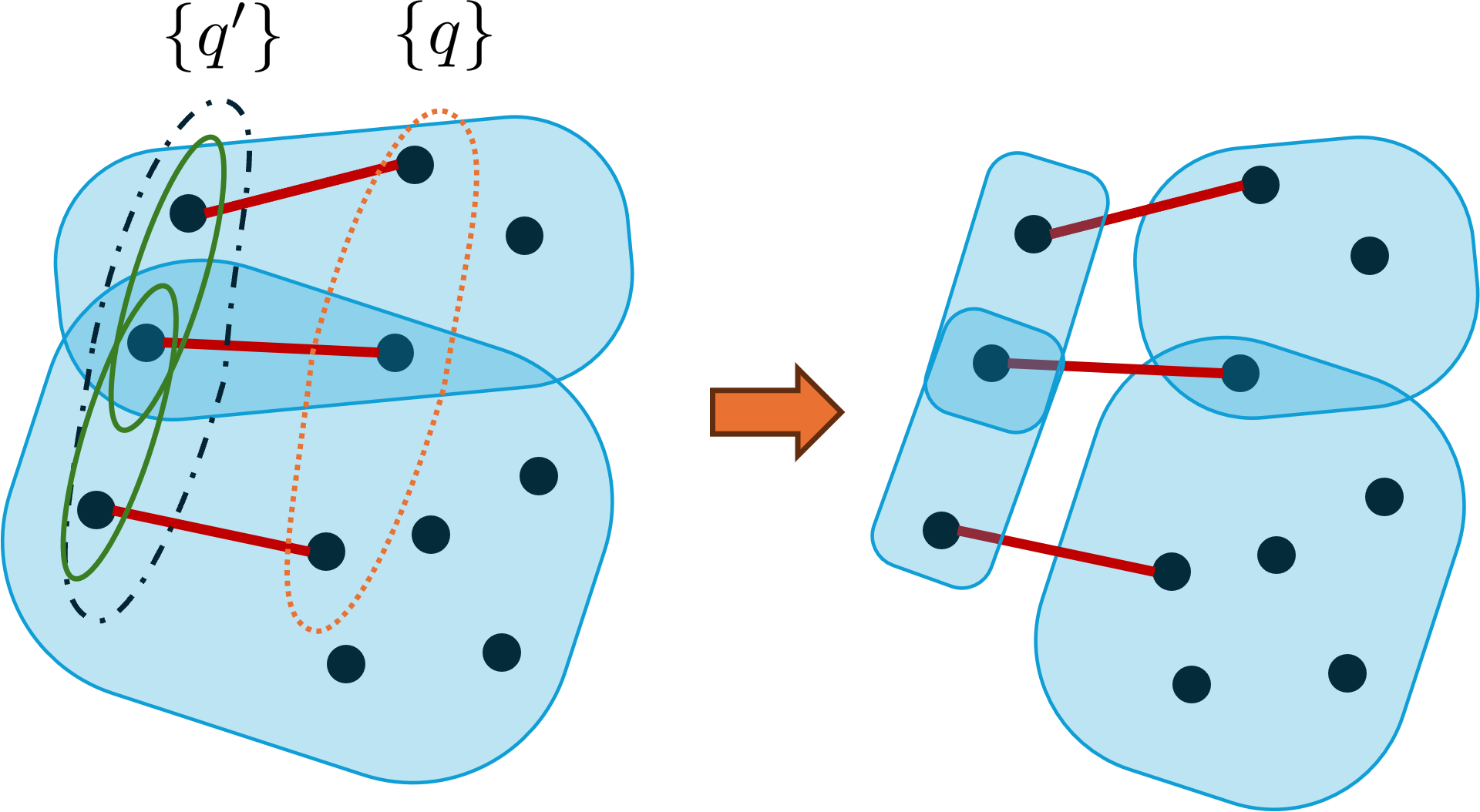}
    \caption{Left: The support of X checks (blue) have been expanded to $\{q'\}$ by concatenating qubits $\{q\}$ with bit-flip codes. Two XNs are applied to qubits in the green circles and their weights reduced (right). The weight two ZZ stabilizers on the left is converted to gauge operators on the right.}
    \label{fig:wt_reduce_pf}
\end{figure}

\begin{lemma}
    The degree of a node that corresponds to a physical qubit in the Tanner graph can be reduced to $q_Z''\leq q_Z$ by shifting at least one of the $Z$-checks supported on $q$ to $q'$ in step (4) of the iterative algorithm. 
\end{lemma}
\begin{proof}
    For each $q\in \mathcal{L}_Z$, it has one more than the degree cap $q_Z$. However, we can use these new weight-$2$ gauge operators to ``shift'' some of the $Z$-checks to $q'$s via multiplication. In other words, we slightly redefine the $Z$-checks by multiplying some of them on $q$ with the weight-$2$ checks between $q$ and $q'$. Note that the newly added $q'$ have $Z$-check degree $1$ since these weight-$2$ $Z$-checks are the only ones supported there. For each $q$, we shift one of the original $Z$-checks to the new site, then all qubits now have check degree $\leq q_Z$. One can also shift more, say $\floor{q_Z/2}$ then this reduces the degree as long as $q_Z\geq 2$.
\end{proof}

As we only redefined the checks, there is no impact on the code distance. So we have finally increased $X$-distance for one logical qubit by at least $1$ while preserving the other minimal distances and sparsity condition. 

Repeating this for all $k$ logical operators, we now have increased the $X$-distance by $1$ and have done nothing to the $Z$-distance. One can then exchange the roles of $X$ and $Z$ to increase the $Z$-distance. This completes the proof of Theorem~\ref{thm:1}.

One important factor to note is that to increase the $Z$-distance, we now have to find the support of some $\bar{X}_0$, which has been altered compared to the original seed code. Fortunately, this is efficiently tractable. 

First, we identify a bare $X$-logical after steps 1-4. One representative is the operator we obtained just after concatenation, when $\bar{X}_0$ is increased by the amount $c=|\mathrm{supp}(\bar{X}_0))\cap \mathrm{supp}(\bar{Z}_0^{(j)})|$. This $\bar{X}_0'$ is a bare logical $X$-operator because any pairs of qubits that pushed through $XN$ do not activate the gauge qubit. Neither does any single qubit flow through $XN$. 

However, the intersection of support $c'=|\mathrm{supp}(\bar{X}_0'))\cap \mathrm{supp}(\bar{Z}_0^{(j)})|=c$ is unchanged because we keep the original $\bar{Z}_0^{(j)}$, which remained a bare $Z$-logical and they will only intersect on the original qubits. The newly added qubits have no effect on this intersection. This repeats for all $k$ pairs of bare logical operators in step 5. Now setting $\bar{X}_0'\rightarrow \bar{X}_0$ we can repeat the same argument but for the phase-flip concatenation steps in step 6. 

\begin{lemma}\label{lemma:baredist}
    At each iteration 2-5, the bare logical distance is increased by at most $c$ where $c\geq 1$ is the number of intersections between $\bar{X}_0$ and $\bar{Z}_0$. 
\end{lemma} 

Furthermore, the sites of these bare logical operators are exactly and efficiently tractable because we can keep track of which $c$ sites have been added to the support of each operator.

Using this algorithm iteratively, we can then grow a sparse code efficiently with a distance and sparsity guarantee by repeating steps 2-6. The operation can also be converted to a transformation on check matrices using the usual symplectic representation.

\section{Proof of Theorem 2}\label{app:b}

We first determine how many physical qubits we need to add to increase code distance by $1$, in the worst case. As we started with a finite size seed code, any overlap between bare logical operators must be bounded by some constant $c$. To add to the distance of the first set of logical operators, from Lemma~\ref{lemma:baredist} we need to add $2d+c$ sites. As the weight of the bare operator for the 2nd logical qubit could have been increased by $c$ during this process, to increase the distance of the 2nd logical qubit we now need to add at worst $2(d+c)+c$ qubits. Repeating this $k$ times, we have 
\begin{align}
    \Delta n &= \sum_{i=0}^{k-1}2[(d+ic)+c] \nonumber\\
    &\leq 2k(d+c)+ck(k-1)\nonumber\\ 
    &\leq ck^2+2kd+kc
\end{align}

Now imagine iterating our algorithm $D$ times. For each iteration, we would replace $d$ by the bare operator distance of the base code, which is $d_1\leq kc+d_0$. At the $j$-th iteration, the bare operator weight that we are tracking would be upper bounded by $d_j\leq kcj+d_0$. Hence to increase the minimum distance by at least $D$, the number of qubits we will add is
\begin{align}
    N&\leq \sum_{j=1}^{D} (ck^2+2kd_j+kc) \nonumber\\
    &= ck^2D +ckD+kD(1+D) \nonumber\\
    &= kD^2+kD(c+1)+ck^2D.
\end{align}

The final code is a $[[N+n,k,d+D]]$. In the asymptotic limit where $D\gg d, N\gg n$, keeping only leading order terms we have a $[[n'=ckD^2, k'=k, d'=D]]$, which completes the proof with relabelling.    

Another limit to notice is the large $k$ limit where $D$ is of a similar size, then we see that there is a contribution $k^2D$ can contribute about equally as $kD^2$.
\end{document}